\documentclass[twoside,leqno]{article}

\newif\ifSODA
\SODAfalse 

\newif\ifArxiv
\ifSODA
\Arxivfalse
\else
\Arxivtrue
\fi

\ifSODA
\usepackage[letterpaper]{geometry}
\usepackage{siamproceedings}
\fi

\ifArxiv
\usepackage[letterpaper,margin=1in]{geometry}
\fi

\usepackage[T1]{fontenc}
\usepackage{amsfonts}
\usepackage{graphicx}
\usepackage{epstopdf}
\usepackage{enumitem}
\usepackage{algorithmic}

\usepackage{graphicx} 

\usepackage
{
        amssymb,
        amsmath,
}

\ifArxiv
\usepackage{amsthm}
\usepackage{times}
\fi

\usepackage
{
        nicefrac,
        xcolor,
}

\usepackage{tikz}
\usetikzlibrary{arrows.meta}

\ifArxiv
\usepackage[unicode,colorlinks=true,citecolor=black,filecolor=black,linkcolor=black,urlcolor=black, pdfpagelabels, plainpages=false]{hyperref}
\fi

\usepackage{bookmark}

\newcommand{\prt}{\operatorname{\vdash}}
\DeclareMathOperator {\Aut}  {Aut}
\DeclareMathOperator {\Tran}  {Tran}
\DeclareMathOperator {\cnt}  {count}
\newcommand {\bell}    {{\cal{B}}}
\newcommand {\abell}    {{\cal{A}}}
\newcommand{\Erangle}{\rangle_{\mathbb{E}}}

\newcommand{\Ebigodot}{%
  \mathop{%
    \tikz[baseline=(X.base)] \node[draw,thick,circle,inner sep=3pt] (X)
      {\scalebox{0.8}{$\mathbb{E}$}};%
  }%
}

\DeclareMathOperator {\Var}  {Var}

\newcommand {\roundup}   [1] {{\lceil {#1} \rceil}}

\newcommand {\brc}   [1] {\left(#1\right)}
\newcommand {\Exp}       {\mathbb{E}}

\newcommand {\Prob}  [1] {\Pr \brc{#1 }}

\newcommand {\EE}    [2] {\Exp_{#1}\left[#2\right]}

\newcommand {\bbN}    {\mathbb{N}}

\newcommand {\bbR}    {\mathbb{R}}

\newcommand {\calS}    {{\cal{S}}}
\newcommand {\calP}    {{\cal{P}}}

\newcommand {\tensor}   {\otimes}

\DeclareMathOperator {\cost}{cost}

\ifSODA
\newsiamremark{remark}{Remark}
\newsiamthm{claim}{Claim}
\else
\newtheorem{theorem}{Theorem}[section]
\newtheorem{definition}[theorem]{Definition}
\newtheorem{lemma}[theorem]{Lemma}
\newtheorem{remark}[theorem]{Remark}
\newtheorem{claim}[theorem]{Claim}
\newtheorem{corollary}[theorem]{Corollary}
\fi

\author{}
\begin{document}

\ifSODA
\newcommand\relatedversion{}
\renewcommand\relatedversion{\thanks{The full version of the paper can be accessed at \protect\url{https://arxiv.org/abs/0000.00000}}} 
\fi

\ifSODA
\title{\Large Hardness of Approximation for Shortest Path with Vector Costs}
    \author{Charlie Carlson\thanks{State University of New York at Buffalo (\email{cc387@buffalo.edu}).}
    \and Yury Makarychev \thanks{Toyota Technological Institute at Chicago (\email{yury@ttic.edu}). Supported by NSF grants CCF-1955173 and ECCS-2216899.}    
    \and Ron Mosenzon\thanks{Toyota Technological Institute at Chicago
  (\email{ron.mosenzon@ttic.edu}). Supported in part by NSF grant CCF-2402283.}}
\date{}
\else
\title{\Large Hardness of Approximation for Shortest Path with Vector Costs\footnote{The conference version of this paper will appear in the Proceedings of SODA 2026.}}
    \author{Charlie Carlson\thanks{State University of New York at Buffalo ({cc387@buffalo.edu}).}
    \and Yury Makarychev \thanks{Toyota Technological Institute at Chicago ({yury@ttic.edu}). Supported by NSF grants CCF-1955173 and ECCS-2216899.}    
    \and Ron Mosenzon\thanks{Toyota Technological Institute at Chicago
  ({ron.mosenzon@ttic.edu}). Supported in part by NSF grant CCF-2402283.}}
\date{}  
\fi

\maketitle

\ifSODA
\fancyfoot[R]{\scriptsize{Copyright \textcopyright\ 2026 by SIAM\\
Unauthorized reproduction of this article is prohibited}}
\fi 

\begin{abstract}
We obtain hardness of approximation results for the $\ell_p$-Shortest Path problem, a variant of the classic Shortest Path problem with vector costs. For every integer $p \in [2,\infty)$, we show a hardness of $\Omega(p(\log n / \log^2\log n)^{1-1/p})$ for both polynomial- and quasi-polynomial-time approximation algorithms. This nearly matches the approximation factor of $O(p(\log n / \log\log n)^{1-1/p})$ achieved by a quasi-polynomial-time algorithm of Makarychev, Ovsiankin, and Tani (ICALP 2025). No hardness of approximation results were previously known for any $p < \infty$. We also present results for the case where $p$ is a function of $n$.

For $p = \infty$, we establish a hardness of $\tilde\Omega(\log^2 n)$, improving upon the previous $\tilde\Omega(\log n)$ hardness result. Our result nearly matches the $O(\log^2 n)$ approximation guarantee of the quasi-polynomial-time algorithm by Li, Xu, and Zhang (ICALP 2025).

Finally, we present asymptotic bounds on higher-order Bell numbers, which might be of independent interest.
\end{abstract}
\section{Introduction}
In this paper, we prove an inapproximability result for the $\ell_p$-Shortest Path problem. The $\ell_p$-Shortest Path Problem is an analogue of the classical $s$-$t$ path problem, in which edge costs are vectors rather than scalars.
\begin{definition}
We are given a directed graph $G = (V, E)$, vector edge costs $c(e)\in\mathbb{R}_{\geq 0}^\ell$ with non-negative coordinates, two vertices $s$ and $t$, and a parameter $p\geq 1$. Define a cost vector for a path $P$ as $\cost(P) = \sum_{e\in P} c(e)$. The $\ell_p$-Shortest Path problem with a parameter $p \in [1,\infty]$  asks to find an $s$-$t$ path $P$ that minimizes $\|\cost(P)\|_p$. 
\end{definition}

The $\ell_p$-Shortest Path problem generalizes the classical Shortest Path problem to account for fairness and robustness constraints. The coordinates of the cost vector represent costs incurred by different parties (modeling fairness), across varying scenarios (capturing robustness), or for different resources (reflecting a traditional multi-objective optimization setting). The parameter $p$ governs how these individual costs are aggregated. As $p \to \infty$, the objective shifts toward minimizing the worst-case (maximum) cost, thereby promoting fair and robust solutions. The $\ell_p$-Shortest Path problem arises in applications ranging from physical networks -- such as transportation and communication systems -- to abstract configuration spaces. Reflecting the growing interest in algorithmic fairness, this problem has recently received renewed attention (see \cite{ABV21,CJMMS23,CMV22,GSV21,MV21} for recent papers modeling fairness constraints using a similar multiobjective optimization approach).
Note that when $\ell=1$ or $p = 1$, the problem is equivalent to the standard shortest path problem.

The problem was first studied for $p=\infty$ under the name of Robust Shortest Path. 
Papadimitriou and Yannakakis showed that the problem admits an FPTAS when the number of coordinates $\ell$ is a fixed constant~\cite{PY2000} (see also a related paper by Aissi, Bazgan, and Vanderpooten~\cite{aissi2007approximation}). Kasperski and Zieliński~\cite{kasperski2009approximability} proved that the $\ell_\infty$-Shortest Path does not admit an 
$\Omega(\log^{1-\varepsilon} n)$-approximation assuming $NP \nsubseteq \mathrm{ZPTIME}(n^{\operatorname{polylog} n})$ (this hardness result was slightly improved to $\Omega(\log n/\log\log n)$ by Makarychev, Ovsiankin, and Tani~\cite{MOT24}). Later, Kasperski and Zieliński also
presented an
$O(\sqrt{n \log \ell / \log \log \ell})$-approximation algorithm, which remains the best known true-polynomial-time algorithm for general instances of $\ell_\infty$-Shortest Path~\cite{kasperski2018approximating}. However, a significantly better approximation is achievable in quasi-polynomial time. Li, Xu, and Zhang~\cite{li2024polylogarithmic} presented an $O(\log n \log \ell)$-approximation algorithm for $p=\infty$ in quasi-polynomial time $n^{O(\log n)}$. Further, their algorithm runs in true polynomial time for graphs of bounded treewidth, and yields an $O(k \log \ell)$-approximation for series-parallel graphs of order $k$.

Makarychev, Ovsiankin, and Tani studied the problem for arbitrary integers $p\geq 2$. They showed that for every constant $c > 0$, the problem admits a $c p \log^{1-1/p} n$ approximation in quasi-polynomial time $n^{O(\log n)}$ and a $c p \frac{\log^{1-1/p} n}{\log\log n}$ approximation in quasi-polynomial $n^{\operatorname{polylog} n}$ time. Also, for a fixed integer $p\geq 2$, the problem can be approximated within $\bell_k(p)^{1/p} = O(pk^{1-1/p})$ in polynomial time in series-parallel graphs of order $k$, where $\bell_k(p)$ is the Bell number of order $k$ (see Section~\ref{sec:bell} for the definition of higher-order Bell numbers and asymptotic formulas for them). The paper also showed that the variant of $\ell_p$-Shortest Path allowing negative coordinates is as hard as the $\ell_p$-Closest Vector Problem and therefore does not admit a better than $\ell^{c_p/\log\log \ell}$ approximation for every $p\in [1,\infty]$ and some $c_p > 0$.

Prior to our work, all known hardness of approximation results for the $\ell_p$-Shortest Path problem exclusively considered the $\ell_\infty$-norm;  no hardness results were known for any $p < \infty$.

\paragraph{Our Results}
We begin by considering the $\ell_p$-Shortest Path problem for a fixed integer parameter $p \geq 2$.
Leveraging Feige’s $r$-prover system~\cite{Feige}, as formulated by Bhattacharya, Chalermsook, Mehlhorn, and Neumann~\cite{BCMN}, we show that the $\ell_p$-Shortest Path problem admits no constant-factor approximation.
Moreover, we establish a tight inapproximability result for series-parallel graphs of order $k$.

\begin{theorem}\label{thm:main-pvsnp}
For every integer $p \geq 2$, the $\ell_p$-Shortest Path problem does not admit a constant factor approximation in series-parallel graphs if $P\neq NP$.

Further, for every $\varepsilon > 0$, and integer parameters $p\geq 2$ and $k\geq 1$, the problem does not admit a $\bell_k(p)^{1/p} - \varepsilon$ approximation in series-parallel graphs of order $k$ if $P\neq NP$.
\end{theorem}

To obtain a superconstant lower bound for the problem, we adopt the stronger complexity assumption that $NP \nsubseteq \bigcap_{\delta > 0} \mathrm{DTIME}(2^{n^{\delta}})$, which allows us to apply the reduction from Theorem~\ref{thm:main-pvsnp} with superconstant values of $k$. We get the following result.
\begin{theorem} \label{thm:hardness-of-ell-p-shortest-path}
    The following hardness result holds for some absolute constant $C' \geq 1$ and every $c \geq 1$ if $NP \nsubseteq \bigcap_{\delta > 0} \mathrm{DTIME}(2^{n^{\delta}})$. Assume that integer parameter $p = p(n) \geq 2$ satisfies
    $$p(n) \cdot \log^* p(n) \leq \frac{\log{n}}{C'\cdot c \cdot \log{\log{n}}}$$    %
    for all sufficiently large $n$ and is non-decreasing as a function of $n$; in particular, $p(n)$ might be equal to some constant $p_0 \geq 2$.
    
    Then, there is no deterministic $\alpha$-approximation algorithm for the $\ell_p$-Shortest Path problem with $p = p(n)$ running in quasi-polynomial time $O(2^{(\log n)^c})$, where
$$\alpha = \frac{1}{C' \cdot c} \cdot \min\left\{p(n) \cdot \left(\frac{\log n}{\log^2\log n}\right)^{1-\frac{1}{p(n)}}, \frac{\log n}{\log \log n} \right\}.$$         
In particular, there is no polynomial-time $\alpha$-approximation algorithm.
\end{theorem}
For a fixed parameter $p < \infty$, this theorem yields an $\Omega(p(\log n / \log^2\log n)^{1-1/p})$-hardness, where the implicit constant in the $\Omega$-notation does not depend on $p$.

Theorems~\ref{thm:main-pvsnp} and~\ref{thm:hardness-of-ell-p-shortest-path} do not directly yield strong hardness results for the $\ell_\infty$-variant of the problem. To obtain a hardness result for $\ell_\infty$, we would need to set $p = \log \ell$, where $\ell$ is the dimension of the cost vector.
However, when we change parameters in our reduction to increase the hardness of approximation factor $\alpha$, the dimension $\ell$ increases as an exponential function of $\alpha$. As a result, we can only establish a hardness of $\tilde{\Omega}(\log \ell)$ -- a bound that is already known. To obtain a stronger $\Omega(\log^2 \ell)$ hardness, we apply a related but modified reduction and derive the following result.

\begin{theorem}\label{thm: hardness of ell-infinity shortest path} 
For some absolute constant $C'' \geq 1 $ and all $c >0$, the following hardness of approximation result holds, assuming that $NP \nsubseteq \bigcap_{\delta > 0} \mathrm{DTIME}(2^{n^{\delta}})$.
There is no deterministic $\alpha$-approximation algorithm for the $\ell_\infty$-Shortest Path problem that runs in time $O(2^{(\log n)^c})$, where  
$$\alpha = \left(\frac{\log n}{C'' \cdot c \cdot (\log \log n)^2}\right)^2 = \tilde\Omega(\log^2 n).$$
\end{theorem}

To obtain our results for $p<\infty$, we need asymptotic lower bounds on the higher-order Bell numbers $\bell_k(p)$.
Although introduced by Bell in 1938 -- almost a century ago -- higher-order Bell numbers have received much less attention than the standard, first order ones~\cite{bell1938iterated}.
Recently, Skau and Kristensen obtained an asymptotic formula for Bell numbers $\bell_k(p)$ for fixed $p$ and $k\to \infty$~\cite{skau2019asymptotic}.
Then Makarychev, Ovsiankin, and Tani obtained an asymptotic upper bound on $\bell_k(p)^{1/p}$ for all $k$ and $p$~\cite{MOT24}. We complement their result by presenting matching lower bounds. We believe this result is of independent interest.

\paragraph{Near-optimality of our results}
Let us compare our hardness results with the corresponding approximation results and see that they are nearly optimal.
We begin with series-parallel graphs. We show that $\ell_p$-Shortest Path does not admit a better than $\bell_k(p)^{1/p} - \varepsilon$ approximation in series-parallel graphs of order $k$; this matches the algorithmic upper bound of $\bell_k(p)^{1/p}$ exactly.

Next, we turn to general graphs. Since all known polylogarithmic approximation algorithms for such graphs run in quasi-polynomial time, we compare our hardness results to those. Note that our hardness results also apply to quasi-polynomial-time algorithms.
For fixed $p$, we show that $\ell_p$-Shortest Path does not admit a better than $\Omega(p \log n / \log^2 \log n)^{1 - 1/p}$ approximation. This matches the approximation result by Makarychev, Ovsiankin, and Tani up to a $(\log \log n)^{1 - 1/p}$ factor.
Our hardness result of $\Omega(\log^2 n / \log^4 \log n)$ for $\ell_\infty$ matches the approximation guarantee of Li, Xu, and Zhang up to a $\log^4 \log n$ factor.

To summarize, our results for series-parallel graphs match the known approximation bounds exactly, and our results for general graphs match the best known guarantees up to $\log \log$-factors.

\paragraph{Non-integer values of $p$}
While all of the aforementioned results hold only for integer values of $p$, we emphesize that obtaining any non-trivial positive and negative results for fractional values of $p$ seems very challenging using current techniques. Indeed, even getting any polylogarithmic lower or upper-bounds for e.g. $p=\nicefrac{3}{2}$, is an interesting open problem.

\subsection*{Technical Overview of the Paper}

Let us present a high-level overview of our hardness reductions. 
We begin with a hard instance of the $r$-Hypergraph Label Cover problem. In this problem, we must label the vertices of a given hypergraph $\cal H$ so that loosely speaking, the labelings of all vertices in each hyperedge are consistent. The exact definition of the consistency requirement is unimportant for this informal overview, though it plays a crucial role in the formal proofs. In the yes-case, the consistency requirement is satisfied for all hyperedges; in the no-case, in a $(1-\varepsilon)$-fraction of hyperedges, the labelings of every pair of vertices are inconsistent.

We remark that this problem serves as a natural and effective starting point for our reduction. As we will see shortly, it is easy to define an $\ell_p$-Shortest Path instance on a series-parallel graph of order 1 such that every $s$–$t$ path corresponds to a hypergraph labeling and vice versa. By assigning suitable vector costs, we ensure a substantial gap between the optimal costs of the resulting yes- and no-instances of $\ell_p$-Shortest Path. For $p < \infty$, it is crucial that our reductions use $r$-Hypergraph Label Cover with large constant or superconstant values of $r$. Although we could instead begin with a Label Cover or even an Independent Set instance, such alternatives would yield considerably weaker hardness results even for series-parallel graphs of order 1 and would prevent any significant amplification of the hardness gap.

Now we present our reduction of $r$-Hypergraph Label Cover to $\ell_p$-Shortest Path. For each vertex $u$ of hypergraph $\cal H$, we create a block of parallel edges and connect these blocks serially (in an arbitrary order); see Figure~\ref{fig:basic}. Edge number~$i$ in the block corresponding to vertex~$u$ represents the labeling $u \leftarrow i$. This yields a graph~$G$ in which every $s$-$t$ path selects exactly one edge per block and thereby defines a labeling of~$\cal H$.

Next, we assign $0$-$1$ cost vectors to the edges to ensure that the following property holds for every hyperedge $h$. Consider any pair of vertices $u, v \in h$. The cost vectors for edges in $G$ corresponding to labelings $u \leftarrow i$ and $v \leftarrow j$ have disjoint supports if the labelings are consistent, and resemble random $0$-$1$ vectors with a fixed number of 1s otherwise. We construct the assignment vectors separately for each hyperedge $h$, and then concatenate them across all hyperedges.

Let us now analyze the cost of the obtained instance. In the yes-case, a consistent labeling for $\cal H$ defines an $s$-$t$ path whose edge cost vectors are pairwise orthogonal, so the total $\ell_p$-cost is small. In the no-case, the vectors overlap in many 1-coordinates, leading to a significantly higher $\ell_p$-cost.

This construction yields a hardness of only $\Theta(p / \log p)$ for $\ell_p$-Shortest Path. To amplify this, we apply a \emph{tensoring} operation: given graphs $G_1$ and $G_2$, we define $G_1 \otimes G_2$ by replacing each edge of $G_1$ with a copy of $G_2$. The cost vector of edge $e_1 \otimes e_2$ in $G_1 \otimes G_2$ is the tensor product of the cost vectors of $e_1$ and $e_2$.

Now we consider the hardness instance $G_1 \otimes G_2$ with $G_1 = G_2 = G$. It is straightforward to show that there is a good $s$-$t$ path in the yes-case. Analyzing the no-case, however, is more challenging. If a path $P$ follows the same path in every copy of $G_2$ within $G_1 \otimes G_2$, the analysis reduces to that for $G$. However, an optimal path~$P$ will most likely follow different paths $P_1, \dots, P_n$ in different copies of $G_2$. Each $P_i$ defines its own labeling of~$\cal H$. Nevertheless,  using that every labeling of $\cal H$ is not even consistent in a weak sense, we are able to analyze this case.
We iterate this tensoring procedure to obtain the final instance $G^{\otimes k} = G \otimes \dots \otimes G$, yielding our hardness result for $\ell_p$-Shortest Path with $p < \infty$.

In the reduction construction for $\ell_\infty$ and its analysis, we use a similar setup but with a different set of cost vectors and a distinct analysis. Unlike in the case of $p < \infty$, the parameter $r$ does not have to be large, and we set $r = 2$. We define the graph $G$ as before and then consider its tensor power $G^{\tensor k}$.

In the yes-case, there exists a path $P$ of $\ell_\infty$-cost 1 corresponding to a satisfying labeling of $\cal H$.  
In the no-case, we inductively prove that the cost of every $s$-$t$ path $P$ in $G^{\otimes j}$ is at least  $y_j$, which is defined by the following recurrence for some fixed $C > 0$:  
$$
y_{j+1} \geq y_{j} + \sqrt{y_{j}}/C.
$$
Let us briefly discuss the inductive step. Consider an $s$-$t$ path $P$ in $G^{\tensor(j+1)}$. We think of $G^{\tensor(j+1)}$ as $G^{\tensor j}$ in which every edge is replaced with a copy of $G$. Thus, $P$ is defined by a path $P'$ in $G^{\tensor j}$ and paths $P_e$ in $G$ for all edges $e$ of $P'$. By the induction hypothesis, the cost of $P'$ is at least $y_j$. In our vector system, all cost vectors have only 0–1 coordinates; this means that there are at least $y_j$ edges of $P'$ sharing a single coordinate 1. We focus on these edges and consider the corresponding paths $P_{1}, \dots, P_{y_j}$. As in the case of $p < \infty$, we use the fact that no labeling of $\cal H$ is consistent, even in a weak sense; without delving into details, this means that for some edge $h = (u, v)$ of $\cal H$, the cost vectors of edges corresponding to vertices $u$ and $v$ in paths $P_{e_1}, \dots, P_{y_j}$ behave like samples of independent Bernoulli 0–1 random variables (in a formalized sense). In each coordinate, there are two samples for each path (one for $u$ and one for $v$), so there are $2y_j$ samples in total.  

Now we make the following observation.  
Consider $2y_j$ independent Bernoulli 0–1 random variables. Their sum $S$ has expectation $\mu = \frac{1}{2} \cdot 2y_j = y_j$ and standard deviation $\sigma = \sqrt{y_j/2}$. Then $S$ exceeds $\mu + \sigma/10$ with constant probability.  
This means that in every coordinate, the sum of cost vectors for $P_{e_i}$ is at least $y_j + \sqrt{y_j/200}$ with constant probability, and the probability that at least one coordinate reaches $y_j + \sqrt{y_j/200}$ is very high if the number of dimensions is large enough. We therefore obtain that the total $\ell_\infty$-cost of $P_1, \dots, P_{y_j}$, and thus of $P$, is at least  
$$
y_{j} + \sqrt{y_{j}/200} \geq y_{j+1},
$$
as desired. Finally, solving the recurrence for $y_j$, we get $y_j = \Omega(j^2)$. Plugging in appropriately chosen parameters, we obtain a hardness of $\tilde{\Omega}(\log^2 N)$, where $N$ is the size of the resulting instance.

\paragraph{Organization}
In Section~\ref{sec:prelim}, we present the preliminaries, including background on Feige’s $r$-prover system and the $r$-Hypergraph Label Cover problem introduced by Bhattacharya, Chalermsook, Mehlhorn, and Neumann.
In Section~\ref{sec:vec-costs}, we construct a vector system that will serve as the basis for our hardness reductions for $p < \infty$.
Section~\ref{sec:reduction-l2} describes our hardness reduction for $\ell_2$. We begin by showing how to convert instances of $r$-Hypergraph Label Cover into instances of $\ell_2$-Shortest Path on series-parallel graphs of depth 1. This reduction yields a hardness factor of only $\sqrt{2} - \varepsilon$.
To amplify the hardness result, we introduce a tensoring operation on instances of $\ell_p$-Shortest Path in series-parallel graphs. Applying this operation, we obtain a hardness of $\sqrt{k+1} - \varepsilon$ in series-parallel graphs of order $k$, establishing Theorem~\ref{thm:main-pvsnp} for $p = 2$.
After completing the case $p = 2$, we turn to the general case of  $p \geq 2$. In Section~\ref{sec:setup-p-greater-than-2}, we introduce additional definitions and background, including partitions and higher-order Bell numbers.
Then, in Section~\ref{sec:analysis-for-general-p}, we adapt our reduction for $p = 2$ to the general case, thereby proving Theorem~\ref{thm:main-pvsnp}.
In Section~\ref{sec:superconstant}, we strengthen our complexity assumption to $NP \nsubseteq \bigcap_{\delta > 0} \mathrm{DTIME}(2^{n^{\delta}})$ and apply our reduction with  superconstant values of $k$ and $r$, resulting in an $\ell_p$-Shortest Path instance of superpolynomial size. This yields Theorem~\ref{thm:hardness-of-ell-p-shortest-path}.
Section~\ref{sec:infinity} addresses the case of $p = \infty$. Using the same high-level framework but modifying many details, we prove that $\ell_\infty$-Shortest Path does not admit a $\tilde\Omega(\log^2 n)$-approximation, establishing Theorem~\ref{thm: hardness of ell-infinity shortest path}.
Finally, in Section~\ref{sec:bounds-on-bell}, we prove tight upper and lower bounds (up to a constant factor) on $\bell_k(p)^{1/p}$.
Appendix~\ref{sec:Feige} provides further discussion of Feige’s $r$-prover system. Appendix~\ref{sec:Bernoulli} contains a bound on sums of Bernoulli random variables used in Section~\ref{sec:infinity}.
\section{Preliminaries}\label{sec:prelim}
\subsection{Hypergraph Label Cover}
Our hardness result is based on a reduction from the $r$-Hypergraph Label Cover, which was introduced by Bhattacharya, Chalermsook, Mehlhorn, and Neumann~\cite{BCMN} (restated).

\begin{definition}[$r$-Hypergraph Label Cover]\label{def: hypergraph label cover}
We are given an $r$-partite hypergraph $\mathcal{H} = (V, \mathcal{E})$ on a set of vertices $V = \bigcup_{j=1}^{r} V_j$ with hyperedges of size $r$. Each hyperedge $h \in \mathcal{E}$ contains exactly one vertex from each $V_i$. Additionally, we are given sets of labels $L$ and colors $C$. Finally, for every hyperedge $h \in \mathcal{E}$ and vertex $u\in h$, we are given a map $\pi_h^u : L \to C$ that maps labels to colors. 

Consider an assignment/labelling $\sigma: V \to L$. We say that $u$ has color $\pi_h^u(\sigma(u))$ in $h$ w.r.t.\ $\sigma$. Note that the color of $u$ may be different in different hyperedges $h$.

An assignment $\sigma$ satisfies a hyperedge $h$ if all the vertices $u\in h$ have the same color in $h$ w.r.t.\  $\sigma$. An assignment $\sigma$ weakly satisfies a hyperedge $h$ if at least two distinct vertices $u', u''\in h$ share the same color in $h$,
$\pi_h^{u'}(\sigma(u')) = \pi_h^{u''}(\sigma(u''))$
\end{definition}

Bhattacharya, Chalermsook, Mehlhorn, and Neumann~\cite{BCMN} show that Feige's $r$-prover system~\cite{Feige} can be cast as an instance of $r$-Hypergraph Label Cover. Therefore, it is NP-hard to distinguish between (a) completely satisfiable instances and (b) instances in which every solution satisfies at most a $\varepsilon$-fraction of the constraints \cite[Theorem 3]{BCMN}.

In our paper, we will need a hardness result with parameters different from those used in \cite{BCMN}. Namely, the hardness reduction for $r$-Hypergraph Label Cover from Theorem 3 of \cite{BCMN} uses Feige's $r$-prover system in which the number $\ell$ of clauses selected by the verifier is linear in the number of provers $r$ (denoted $k$ in Feige's paper). However, Feige showed how to construct an $r$-prover system with $\ell=O(\log r)$; see Section 6 in \cite{Feige} for details.
To achieve our results, we will need a version of the reduction from \cite{BCMN} altered to use this parameterization. We present the statement of this altered reduction in the following theorem.

\begin{theorem}[Hardness reduction for $r$-Hypergraph Label Cover, based on the reduction from \cite{BCMN}]\label{thm:hardness-r-hypergraph-label-cover}
There is an algorithm (reduction) that, given an $n$-variable 3-SAT formula, a parameter $r\geq 2$, and a parameter $0 < \varepsilon < 1$, outputs an instance $\cal H$ of $r$-Hypergraph Label Cover such that

\begin{itemize}
\item \textbf{Yes-Instance.} If the input formula is satisfiable, then the instance $\cal H$ is completely satisfiable.
\item \textbf{No-Instance.} If the input formula is unsatisfiable, then every assignment for $\cal H$ weakly satisfies at most an $\varepsilon$-fraction of the hyperedges.
\end{itemize}
The hypergraph in the produced instance of $r$-Hypergraph Label Cover contains $n^{O\left(\log \frac{r}{\varepsilon}\right)}$ vertices and hyperedges.
The sets of labels and of colors are both of cardinality at most $2^{O\left(\log \frac{r}{\varepsilon}\right)}$.
The running time of the algorithm is $n^{O\left(\log \frac{r}{\varepsilon}\right)}$.
\end{theorem}
See Appendix~\ref{sec:Feige} for further details on Theorem~\ref{thm:hardness-r-hypergraph-label-cover}.

\subsection{Series-Parallel Graphs}
Recall the definition of series-parallel graphs with source $s$ and sink $t$. We define series-parallel graphs recursively and for convenience we allow parallel edges in series-parallel graphs. 
A pair of vertices $s$ and $t$ connected by a single or multiple edges is a series-parallel graph of order 0.
Given series parallel graphs $G_1, \dots, G_t$ of order $k$, we consider their parallel and series compositions.
The parallel composition of the graphs has depth $k$; the series composition of the graphs has depth $k+1$.
Note that in our definition, we count only the number of series compositions, not the number of parallel ones.

\subsection{Inner Products}
We first present our results for the Euclidean norm. We will need two variants of the inner product and the $\ell_2$ norm. First, we use the standard, unscaled inner product on $\mathbb{R}^D$: $\langle v_1, v_2\rangle = \sum_{j=1}^D v_1(j)\cdot v_2(j)$, where $v_i(j)$ denotes the $j$-th coordinate of vector $v_i$. Accordingly, $\|v\|_2 = \langle v, v\rangle$.

Then, it will be convenient to use a scaled variant of the inner product in some spaces. In such cases, we fix a scaling parameter $\alpha$ for the space and write $\langle v_1, v_2\Erangle = \alpha \sum_{j=1}^D v_1(j)\cdot v_2(j)$. Correspondingly, we use a scaled $\ell_2$-norm: $\|v\|^2_2 = \langle v, v\Erangle = \alpha\sum_{j=1}^D v(j)^2$. We will never consider both the unscaled and scaled inner product on the same space.

We also consider tensor products of Euclidean spaces. As is standard, $\langle a_1\tensor b_1, a_2\tensor b_2\rangle = \langle a_1, a_2\rangle\cdot \langle b_1, b_2\rangle$. If spaces $U_1$ and $U_2$ have scaling parameters $\alpha_1$ and $\alpha_2$, then we assign $U\tensor V$ the scaling parameter $\alpha_1\alpha_2$ so that
$\langle a_1\tensor b_1, a_2\tensor b_2\Erangle = \langle a_1, a_2\Erangle\cdot \langle b_1, b_2\Erangle$ for $a_i \in U_1$ and $b_j \in U_2$.

\section{Vector Costs}\label{sec:vec-costs}
In this section, we describe a vector system that we will later use to define vector edge costs.
\begin{definition}\label{def : vector system}
Consider a collection of vectors $\calS = \{v_i^c: i\in [r], c\in [q]\}$ (all vectors $v_i^c$ are distinct). Let us say that $c$ is the color of vector $v_i^c$. $\calS$ is an $(r, q, p)$-vector system if it satisfies the following property. Consider $p$ vectors $u_1, \dots, u_p\in S$. The same vector may appear more than once in this list.
\begin{itemize}
    \item  If there are two distinct vectors $u_i \neq u_j$ of the same color in the list, then 
\begin{equation}\label{eq:vs-orthogonal}
\langle u_1,\dots, u_p\rangle = 0
\end{equation}
where $\langle u_1,\dots, u_p\rangle$ denotes the multilinear product of $u_1, \ldots, u_p$ (See Definition \ref{def:multilinear_product}). 
    \item Otherwise, 
\begin{equation}\label{eq:vs-non-orthogonal}
\langle u_1,\dots, u_p\rangle = r^{-t}
\end{equation}    
    where $t$ is the number of distinct vectors among $u_1,\dots, u_p$.
    In particular, for $t=1$, we have 
\begin{equation}
    \label{eq:vs-norm}
\|u\|_p^p = \langle u,\dots, u\rangle = r^{-1}.
\end{equation}
\end{itemize}
\end{definition}

\begin{lemma}\label{lem:vector-system}
Assume that $r$ is a power of 2 and $p\geq 2$ is an integer.
There exists an $(r,q,p)$-vector system with dimension $r^p$ and $q = (r^p-1)/(r-1)$. 
Further, it can be constructed in time polynomial in $q$.
\end{lemma}
\begin{proof}
Fix $P = {\mathbb F}_r^p$.
For each $a = (a_1, \ldots, a_p) \in P \setminus \{(0,\ldots, 0)\}$ and $b \in {\mathbb F}_r$, we associate an affine hyperplane $H(a,b) = \{ z \in P : \langle a, z \rangle = b\}$ in $P$. Note that $(a,b)$ and $(a',b')$ define the same hyperplane $H = H(a,b) = H(a',b')$ if and only if $a' = \lambda a$ and $b' = \lambda b$ for some $\lambda \in \mathbb{F}_r^* \equiv \mathbb{F}_r\setminus \{0\}$.
Consequently, there are $r\cdot |P|/(r-1) = r\cdot (r^p-1)/(r-1) = r\cdot q$ such affine hyperplanes in $P$. 

Fix a direction $a \in P \setminus \{(0,\ldots, 0)\}$. Hyperplanes $H(a,b)$ with $b\in\mathbb{F}_r$ are parallel and partition space $P$: $P = \bigsqcup_{b\in \mathbb{F}} H(a,b)$.
Two directions $a$ and $a'$ define the same class of parallel hyperplanes if and only if vectors $a$ and $a'$ are collinear. 
All affine hyperplanes are divided into disjoint classes of parallel ones. The total number of classes is $q = (r^p-1)/(r-1)$. We refer to these classes as \emph{colors} and observe that there are $r$ affine hyperspaces of each color.

Fix an arbitrary ordering of the color classes and the affine hyperplanes in each color class. 
For each $i \in [r]$ and $c \in [q]$ we define a vector
\[
    v^c_i := \begin{cases}
        1/r & \text{if } z \in H^c_i \text{ and} \\
        0 & \text{otherwise.}
    \end{cases}
\]
where $H^c_i$ is the $i$th affine hyperplane in the $c$-th color class. 
Let $\mathcal{S} := \{v_i^c : i \in [r], c \in [q]\}$.

We will now show that $\mathcal{S}$ is a $(r,q,p)$-vector system by showing it satisfies both properties in the previous definition.
Fix a set of $p$ vectors $u_1, \ldots, u_p$ in $\mathcal{S}$.
Suppose there exist $u_i \neq u_j$ that correspond to the same color class.
Then for each coordinate $w$, $u_j(w) \cdot u_i(w) = 0$ and thus,
$$
    \langle u_1, \ldots, u_p \rangle = 0.
$$
Now suppose otherwise, and let $t$ be the number of distinct vectors among $u_1, \ldots, u_p$ which correspond to distinct, non-parallel affine hyperplanes $H_1, \ldots, H_t$.
Then by definition, 
$$
\langle u_1, \ldots, u_p \rangle = \frac{\left| H_1 \cap \ldots \cap H_t \right|}{r^p} = \frac{r^{p-t}}{r^p} = r^{-t}.
$$ 

Finally, we show that we can construct $\mathcal{S}$ in time polynomial in $q$. 
We can enumerate every affine hyperplane and assign it to the appropriate color class in time $O(rq)$ by enumerating every $a \in P \setminus \{(0,\ldots, 0)\}$ and $b \in {\mathbb F}_r$.
With the set of all color classes and affine hyperplanes in some fixed order, we can construct each vector in $\mathcal{S}$ in time $O(r^p)$. 
Since $|\mathcal{S}| = rq$, this completes the proof.
\end{proof}




\section{Hardness Reduction for \texorpdfstring{$\ell_2$}{l2}}
\label{sec:reduction-l2}
In this section, we will present our hardness result for $p=2$. In subsequent sections, we will generalize it to an arbitrary integer $p \geq 2$.
We will start with describing a basic reduction that lies in the core of our argument. This reduction alone shows only that the $\ell_2$-Shortest Path does not admit a $\sqrt{2 - \varepsilon}$-approximation in series-parallel graphs of order 1. However, later we will amplify this result by ``tensoring'' the instance of $\ell_2$-Shortest Path given by the basic reduction. 

\ifSODA
\pagebreak
\fi

\subsection{Basic Reduction}\label{subsec: basic reduction}
\paragraph{Reduction} Consider an instance $\mathcal{H} = (V, \mathcal{E})$ of $r$-Hypergraph Cover with labels $L$ and colors $C$. 
For every vertex $u$, we construct a block of $|L|$ parallel edges $\{e(u, l)\}_{l\in L}$, indexed by labels $L$. Then we serially compose all the blocks, in an arbitrary order, and obtain graph $G$ (see Figure~\ref{fig:basic}).
Observe that there is a one-to-one correspondence between assignments $\sigma$ for $\cal H$ and $s$-$t$ paths $P$ in $G$: assignment $\sigma$ corresponds to the path $P_\sigma$ formed by edges $\{e(u, \sigma(u)):u\in V\}$.

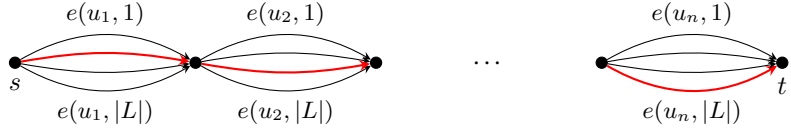
\begin{figure}
\centering
\begin{tikzpicture}[>=stealth, scale=1.2]
    \tikzset{
      vertex/.style={
         circle, draw, fill=black,
         inner sep=1.5pt
      }
    }

    \node[vertex, label=below:{$s$}] (s) at (0,0) {};
    \node[vertex] (a) at (2,0) {};
    \node[vertex] (b) at (4,0) {};

    \node[vertex] (c) at (6.5,0) {};
    \node[vertex, label=below:{$t$}] (t) at (8.5,0) {};

    \draw[->,bend left=30] (s) to node[above] {\small $e(u_1,1)$} (a);
    \draw[->,bend right=10] (s) to (a);
    \draw[->,bend right=30] (s) to node[below] {\small $e(u_1,|L|)$}(a);
    \draw[->,bend left=10,red,thick] (s) to (a);

    \draw[->,bend left=30] (a) to node[above] {\small $e(u_2,1)$} (b);
    \draw[->,bend left=10] (a) to (b);
    \draw[->,bend right=30] (a) to node[below] {\small $e(u_2,|L|)$}(b);
    \draw[->,bend right=10,red,thick] (a) to (b);

    \draw[->,bend left=30] (c) to node[above] {\small $e(u_n,1)$} (t);
    \draw[->,bend left=10] (c) to (t);
    \draw[->,bend right=30,red,thick] (c) to node[below] {\small\color{black} $e(u_n,|L|)$}(t);
    \draw[->,bend right=10] (c) to (t);
   
    \path (b) -- (c) node[midway] {\dots};
\end{tikzpicture}
\caption{This figure shows the basic reduction when $|L| = 4$. The red path corresponds to the assignment $\sigma$ with $\sigma(u_1) = 2$, $\sigma(u_2) = 3$, \dots,  $\sigma(u_n) = 4$}
\label{fig:basic}
\end{figure}

Now we are ready to assign vector costs to all the edges.
To construct an 
$(r, \lvert C \rvert, 2)$-vector system, we first construct a $(r, (r^{p'}-1)/(r-1),p')$-vector system $\mathcal{S}'$ using Lemma~\ref{lem:vector-system} where $p' \geq 2$ is the smallest integer such that $r^{p'} - 1 \geq \lvert C \rvert$.
We then let $\mathcal{S} \subseteq \mathcal{S}'$ be the subset of vectors that correspond only to the first $\lvert C \rvert$ color classes.

Let $d_0 = r^{p'}$ be the dimension of the vectors in the vector system $\mathcal{S}$ and $m = |\mathcal{E}|$.
Define $d=d_0 \cdot m$. All cost vectors will lie in $\mathbb{R}^d$. We divide all the coordinates of $\mathbb{R}^d$ into $m = |\mathcal{E}|$ groups with $d_0$ coordinates in each group. The groups are indexed by the hyperedges $h \in \mathcal{E}$ of $\mathcal{H}$. We define each cost vector $c(e(u,l))$ by specifying its restriction $c^h(e(u,l))$ to each of the groups $h$. In other words, we define $c^h(e(u,l))\in {\mathbb R}^{d_0}$ and then let 
$$c(e(u,l)) = \bigoplus_{h\in \mathcal{E}} c^h(e(u,l))\in {\mathbb R}^d.$$
Consider a hyperedge $h$. For $u\in h$,
let  $c^h(e(u,l)) = v_j^c$ where $j$ is such that $u\in V_j$ and $c = \pi_h^u(l)$.
 For $u\notin h$, $c^h(e(u,l)) = 0$.
 To simplify formulas, we use the rescaled inner product in $\mathbb{R}^d$ (as described in Section~\ref{sec:multilinear})
 $$\langle c(e_1), c(e_2)\Erangle = \frac{1}{m} \sum_{h\in \mathcal{E}} \langle c^h(e_1), c^h(e_2)\rangle =  {\mathbb{E}_h}\langle c^h(e_1), c^h(e_2)\rangle$$
 and
  $$\|c(e)\|_2^2 = {\mathbb{E}_h}[\|c^h(e)\|_2^2].$$
In the expectation $\mathbb{E}_h[\cdot]$, $h$ is sampled uniformly at random from $\cal E$.

\paragraph{Size Analysis}
The graph $G$ contains one block for every vertex of the instance $\mathcal{H}$ of $r$-Hypergraph Label Cover. Furthermore, each block contains $|L|$ edges, where $L$ is the set of labels in $\mathcal{H}$.
Lastly, the dimension of the cost vectors in the resulting instance of $\ell_p$-shortest paths is $d_0 \cdot |\mathcal{E}|$, where $d_0$ is the dimension of the vector system, and $\mathcal{E}$ is the set of hyperedges in $\mathcal{H}$.
To construct a $(r,|C|,2)$-vector system using Lemma~\ref{lem:vector-system}, we need to use this lemma with parameter $p'\geq 2$ such that $(r^{p'}-1)/(r-1) \geq |C|$. By choosing the smallest integer $p'$ satisfying these conditions, we get that the dimension $d_0$ of the vector system is upper bounded by $|C| \cdot r^2$, and so is the running time required to build the vector system.
Let $|\mathcal{H}|$ denote the maximum between the number of vertices in $\mathcal{H}$, the number of labels in $\mathcal{H}$, the number of colors in $\mathcal{H}$, and the number of hyperedges in $\mathcal{H}$.
Then, $|E(G)| \leq |\mathcal{H}|^2$, and the dimension of the cost vectors in the resulting $\ell_p$-shortest paths instance is upper bounded by $\max\{r^2, |\mathcal{H}|\cdot r\} \cdot |\mathcal{H}|$.

\begin{lemma}\label{lem:l2-basic}
Consider an instance $\mathcal{H}$ of $r$-Hypergraph Cover and the corresponding $\ell_2$-Shortest Path instance $G$. 
\begin{itemize}
    \item If $\cal H$ is a yes-instance, then there is an $s$-$t$ path of squared $\ell_2$-length $1$ in $G$.
    \item If $\cal H$ is a no-instance, then every $s$-$t$ path has squared $\ell_2$-length at least $2-(1/r + 2\varepsilon)$ in $G$.
\end{itemize}
\end{lemma}
\begin{proof}
Consider an assignment $\sigma$ and hyperedge $h = (u_1, \dots, u_r)$ with $u_j \in V_j$. Let $c_j$ be the color of $u_j$ in $h$ w.r.t.\  $\sigma$. Then
$$\Bigl\|\sum_{e\in P_{\sigma}} c^h(e)\Bigr\|_2^2 = \Bigl\|\sum_{j=1}^r c^h(e(u_j, \sigma(u_j)))\Bigr\|_2^2 = \Bigl\|\sum_{j=1}^r v_{j}^{c_j}\Bigr\|_2^2.$$

\noindent If $\sigma$ satisfies $h$, then all colors $c_j$ are equal; let $c= c_1 = \dots = c_r$. By \eqref{eq:vs-orthogonal}, all vectors $v_j^c$ are orthogonal, and by~\eqref{eq:vs-norm}, $\|v_j^c\|^2 = 1/r$. Therefore,
$$\Bigl\|\sum_{e\in P_{\sigma}} c^h(e)\Bigr\|_2^2 = \Bigl\|\sum_{j=1}^r v_j^c\Bigr\|_2^2 = r\cdot \frac{1}{r} = 1.$$

\noindent If $\sigma$ does not weakly satisfy $h$, then all colors $c_j$ are distinct. Then $\langle v_{j}^{c_j}, v_{j'}^{c_{j'}}\rangle = 1/r^2$ by \eqref{eq:vs-non-orthogonal}. Therefore, 
$$\Bigl\|\sum_{e\in P_{\sigma}} c^h(e)\Bigr\|_2^2 = \Bigl\|\sum_{j=1}^r v_j^{c_j}\Bigr\|_2^2 = 
\sum_{j=1}^r \|v_j^{c_j}\|_2^2 + \sum_{j\neq j'} \langle v_{j}^{c_{j}}, v_{j'}^{c_{j'}}\rangle
= r\cdot \frac1r + r(r-1) \cdot \frac{1}{r^2} = 2 - \frac1r.
$$

\noindent In the \textbf{yes-case}, there is an assignment $\sigma$ satisfying all $h \in \cal H$. We have,
\begin{equation}    \label{eq:base-case-yes}
\|\cost(P_\sigma)\|_2^2 = \Bigl\|\sum_{e\in P_\sigma}c(e)\Bigr\|_2^2 = {\mathbb{E}_h}\Bigl\|\sum_{e\in P_{\sigma}} c^h(e)\Bigr\|_2^2 = 1.
\end{equation}

\noindent In the \textbf{no-case}, for every assignment $\sigma$, at least a $1-\varepsilon$ fraction of the hyperedges is not even weakly satisfied. Thus,
\begin{equation}\label{eq:base-case-no}
\|\cost(P_\sigma)\|_2^2 = \Bigl\|\sum_{e\in P_\sigma}c(e)\Bigr\|_2^2 = {\mathbb{E}_h}\Bigl\|\sum_{e\in P_{\sigma}} c^h(e)\Bigr\|_2^2  \geq (1-\varepsilon) \left(2 - \frac1r\right) \geq 2 - 2\varepsilon - \frac1r.
\end{equation}
\end{proof}

\subsection{Series-parallel graphs of order 2 and higher}\label{sec: l2-general-reduction}
We next extend Lemma~\ref{lem:l2-basic} to series-parallel graphs of higher depths and thus, obtain stronger hardness results for $\ell_2$-Shortest Path. We now describe the main technical tool we employ in this section.

\paragraph{Tensoring}
Given two series-parallel graphs with vector edge costs $G_1=(V_1, E_1)$ and $G_2 = (V_2, E_2)$, we define a new graph $G_1\tensor G_2$. We start with $G_1$ and replace each edge of $G_1$ with a copy of $G_2$. The edges of $G_1\tensor G_2$ are indexed by pairs of edges $(e_1, e_2)$ with $e_1 \in E_1$ and $e_2\in E_2$ in the natural way. We define the cost $c(e_1, e_2)$ of edge $(e_1, e_2)$ as $c_1(e_1)\tensor c_2(e_2)$, where $c_1(e_1)$ and $c_2(e_2)$ are the costs of $e_1$ and $e_2$ in $G_1$ and $G_2$, respectively. 

Now define $G^{\tensor k} = \underbrace{G\tensor \cdots \tensor G}_{k \text{ times}}$. If $G$ is a series-parallel graph of order 1 (as in our reduction), then $G^{\tensor k}$ is of order $k$. To simplify the notation, we also define $G^{\tensor 0}$ as a single edge $(s,t)$ of cost $1\in \mathbb{R}^1$. Note that then $G^{\tensor 1} \equiv G =  G\tensor G^{\tensor 0}.$

Consider an $s$-$t$ path $P$ in $G_1\tensor G_2$. Path $P$ defines a ``projected" path $P' = (e_1, \dots, e_t)$ in $G_1$. For every edge $e_i\in P'$, let $P_i$ be the segment of $P$ inside the copy of $G_2$ that replaced edge $e_i$ in $G_1$. Then $P$ is uniquely specified by $(e_1, P_1),\dots, (e_t, P_t)$. If all $P_i$ are equal: $P'' = P_1 = \cdots = P_t$, we write $P = P' \tensor P''$. 

\paragraph{Colorful Hyperedges}
We will also need the following definition of a colorful hyperedge $h$ w.r.t.\ assignments $\sigma_1, \dots, \sigma_p$. Loosely speaking, $h$ is colorful w.r.t.\ $\sigma_1, \dots, \sigma_p$ if not only does no single assignment $\sigma_i$ weakly satisfy $h$ but also, in a precise sense explained below, all the assignments $\sigma_1,\dots, \sigma_p$ together do not weakly satisfy $h$.
\begin{definition}\label{def:colorful} Consider arbitrary assignments $\sigma_1, \dots, \sigma_p$. Let us say that $h$ is colorful w.r.t.\ if $\pi_h^u(\sigma_i(u)) \neq \pi_h^v(\sigma_j(v))$ for all $i,j\in [p]$ and all distinct $u,v\in h$. 
\end{definition}
\noindent Note that $h$ is colorful w.r.t.\ one assignment $\sigma$ if and only if $\sigma$ does not weakly satisfy $h$.

\begin{lemma} \label{lem:colorful} Consider a no-instance $\cal H$ of $r$-Hypergraph Cover and assignments $\sigma_1,\dots, \sigma_p$. Then at least a $1 - p^2 \varepsilon$ fraction of hyperedges $h\in\cal E$ are colorful.
\end{lemma}
\begin{proof} Let $\delta$ be the fraction of colorful hyperedges. Define a new assignment $\sigma$ by choosing the value of $\sigma(u)$ uniformly at random among $\sigma_1(u), \dots, \sigma_p(u)$ for every $u$; the choices for all vertices $u$ are independent. If $h$ is not colorful then $\pi_h^u(\sigma_i(u)) = \pi_h^v(\sigma_j(v))$ for some distinct $u$ and $v$ and some $i$ and $j$. Then with probability at least $1/p^2$, we have $\sigma(u) = \sigma_i(u)$ and $\sigma(v) = \sigma_j(v)$. In that case, $\sigma$ weakly satisfies $h$. We get that $\sigma$ weakly satisfies at least a $\delta/p^2$ fraction of the hyperedges, in expectation. However, $\cal H$ is a no-instance and thus no assignment weakly satisfies more than an $\varepsilon$-fraction of the vertices. We get that $\delta/p^2 \leq \varepsilon^2$ and $\delta \leq p^2 \varepsilon$, as required.
\end{proof}

We are ready to state the main results of this section.

\begin{theorem} \label{thm:main} Consider an instance $\mathcal{H}$ of $r$-Hypergraph Cover and the corresponding $\ell_2$-Shortest Path instance $G$. 
\begin{itemize}
    \item If $\cal H$ is a yes-instance, then there is an $s$-$t$ path of squared $\ell_2$-length $1$ in $G^{\tensor k}$.
    \item  If $\cal H$ is a no-instance, then every $s$-$t$ path has squared $\ell_2$-length at least $\alpha_k\cdot (k+1)$ in $G^{\tensor k}$, where $\alpha_k = (1 - (1/r + 4\varepsilon))^k = 1 - O(k(1/r + \varepsilon))$.
\end{itemize}
\end{theorem}

By choosing parameters $r$ sufficiently large and $\varepsilon > 0$ sufficiently small in Theorem~\ref{thm:hardness-r-hypergraph-label-cover}, we get the following corollary that yields Theorem~\ref{thm:main-pvsnp} for $p=2$.

\begin{corollary} Assume $P\neq NP$. For every $k\geq 1$, the $\ell_2$-Shortest Path problem does not admit a $\sqrt{k + 1} - \varepsilon$ approximation in series-parallel graphs of order $k$ (for every $\varepsilon > 0$).  Consequently, it
does not admit a constant factor approximation in series-parallel graphs.
\end{corollary}

\begin{proof}[Proof of Theorem~\ref{thm:main}]
Assume first that $\cal H$ is a yes-instance. Let $P = P_{\sigma}$ be an $s$-$t$ path in $G$ of squared length $r$.
Consider the path $P^{\tensor k}$ in $G^{\tensor k}$. It consists of edges $\{(e_1,\dots, e_k):e_1,\dots, e_k \in P\}$.
Therefore, $\cost(P^{\tensor k}) = \cost(P)^{\tensor k}$. Since $\|\cost(P)\|^2 = 1$, we get $\|\cost(P^{\tensor k})\|^2 = 1$, as required.

Now, consider the case when $\cal H$ is a no-instance.
We prove the following by induction on $k$:
\begin{enumerate}
\item $\|\cost(P)\|_2^2 \geq \alpha_k \cdot(k+1)$ for every $s$-$t$ path $P$ in $G^{\tensor k}$ \label{enum : first : thm : main}
\item $\langle \cost(P),\cost(Q)\Erangle \geq \alpha_k$ for every two $s$-$t$ path $P$ and $Q$ in $G^{\tensor k}$. \label{enum : second : thm : main}
\end{enumerate}

\paragraph{Base Case} We proved the first item for $k=1$ in Lemma~\ref{lem:l2-basic}. Let us prove the second item for $k=1$. Consider two paths $P_{\sigma_1}$ and $P_{\sigma_2}$ in $G$. By Lemma~\ref{lem:colorful} (applied with $p=2$), at least a $1-4\varepsilon$ fraction of hyperedges are colorful w.r.t.\  $\sigma_1$ and $\sigma_2$. 
Consider a colorful hyperedge $h=(u_1,\dots, u_r)$. Let $c_1, \dots, c_r$ and $c_1',\dots, c_r'$ be the colors of $u_1,\dots, u_r$ in $h$ w.r.t.\  $\sigma_1$ and $\sigma_2$, respectively. If $i\neq j$, then $c_i \neq c_j'$ and   $\langle v_i^{c_i},v_j^{c_j'}\rangle = 1/r^2$ by~\eqref{eq:vs-non-orthogonal}. If $i = j$, then either (i) $c_i = c_j'$, $v_i^{c_i} = v_j^{c_j'}$, and
$\langle v_i^{c_i}, v_j^{c_j'}\rangle = 1/r$ or (ii)
$c_i\neq c_j'$ and
$\langle v_i^{c_i}, v_j^{c_j'}\rangle = 1/r^2$. In every case, $\langle v_i^{c_i}, v_j^{c_j'}\rangle \geq 1/r^2$ for all $i$ and $j$.
Thus, we have
$$\left\langle\sum_{e_1\in P_{\sigma_1}} c^h(e_1), \sum_{e_2\in P_{\sigma_2}} c^h(e_2)\right\rangle = \left\langle\sum_{i=1}^r v_i^{c_i}, \sum_{j=1}^r v_{j}^{c'_{j}} \right\rangle \geq r^2 \cdot \frac{1}{r^2} = 1.
$$
Hence,
$$\langle \cost(P_1),\cost(P_2)\Erangle \geq 1 - 4\varepsilon \geq \alpha_1.$$

\paragraph{Induction Step}
Consider an $s$-$t$ path $P$ in $G^{\tensor k} = G \tensor G^{\tensor k - 1}$. Write it as $(e_1,P_1),\dots, (e_n, P_n)$, where edges $e_1,\dots, e_n$ form a path $P'=P_{\sigma}$ in $G$, and $P_1, \dots, P_n$ are paths in $G^{\tensor k - 1}$.
We have,
\begin{align*}\|\cost(P)\|^2_2 &= \|\sum c(e_i)\tensor \cost(P_i)\|_2^2 \\&= 
\sum_{i=1}^r \|c(e_i)\|_2^2 \cdot \|\cost(P_i)\|^2_2 +
\sum_{i\neq j} \langle c(e_i), c(e_j)\Erangle \cdot \langle \cost(P_i), \cost(P_j)\Erangle
\end{align*}
where $\cost(P_i)$ is the cost vector of the path $P_i$ in $G^{\tensor k-1}$. Note that
\begin{itemize}
    \item  by~\eqref{eq:vs-non-orthogonal}, $\|c(e_i)\|_2^2=1/r$
    \item by Lemma~\ref{lem:l2-basic},
\begin{align*}
    \sum_{i\neq j} \langle c(e_i), c(e_j)\Erangle &= \sum_{i,j} \langle c(e_i), c(e_j)\Erangle - \sum_i \|c(e_i)\|^2 =\|\cost(P')\|^2 - 1 \\
    &\geq 2-(1/r+2\varepsilon) -1 = 1 -(1/r+2\varepsilon).
\end{align*}
    \item by the induction hypothesis:
\begin{align*}    
    \|\cost(P_i)\|^2_2 &\geq \alpha_{k-1} k\\
    \langle \cost(P_i), \cost(P_j)\Erangle &\geq \alpha_{k-1}.
\end{align*}
\end{itemize}
We get,
$$\|\cost(P)\|^2_2 \geq r \cdot \frac{1}{r}\cdot \alpha_{k-1} k + \alpha_{k-1}\sum_{i\neq j}\langle c(e_i), c(e_j)\Erangle \geq \alpha_{k-1}(k + 1 - (1/r + 2\varepsilon))\geq \alpha_k \cdot(k+1).$$
This completes the proof of Item~\ref{enum : first : thm : main}. 

Now consider two paths $P$ and $Q$. Write them as $(e_1, P_1), \dots, (e_n, P_n)$ and $(e_1', Q_1), \dots, (e_n', Q_n)$. We denote the path formed by edges $e_i$ by $P'$ and by edges $e_i'$ by $Q'$.

\begin{align*}
\langle \cost(P), \cost(Q)\Erangle &= \sum_{i,j} \langle c(e_i), c(e_j')\Erangle\cdot \langle \cost(P_i), \cost(Q_j)\Erangle \geq \alpha_{k-1} \sum_{i,j} \langle c(e_i), c(e_j')\Erangle \\
&=\alpha_{k-1} \cdot \langle \cost(P'), \cost(Q')\Erangle \geq
(1 - 4\varepsilon)\alpha_{k-1} \geq\alpha_k,
\end{align*}
as required to show Item~\ref{enum : second : thm : main}.
\end{proof}

\section{Setup for the \texorpdfstring{$p > 2$}{p > 2} Regime}
\label{sec:setup-p-greater-than-2}
In this section, we present additional definitions that we need to analyze the general case with an arbitrary integer $p > 2$.

\subsection{Partitions and Higher-order Bell Numbers}
\label{sec:bell}
In our reduction for $p > 2$, we make extensive use of labeled partitions and higher-order Bell numbers. We now introduce the necessary definitions.
\begin{definition} We consider the set $[p]$ and its partitions into disjoint, non-empty sets, called parts.  We denote the set of all partitions of $[p]$ by $\calP$. Given a partition $T$ of $[p]$, let $\calP(T)$ be the set of all subdivisions (refinements) of $T$.
\end{definition}
\begin{definition}
A $k$-level partition of $[p]$ is a hierarchical partition of $[p]$ with $k$ levels. Formally, it is a tuple $(T_1,\dots, T_k)$ where $T_1\in \mathcal{P}$ and $T_i \in \mathcal{P}(T_{i-1})$ for $i \in\{2,\dots,k\}$.
\end{definition}
\begin{definition}[see Definition 2.1 in~\cite{MOT24} for discussion]
The order-$k$ Bell number $\bell_k(p)$ is the number of $k$-level hierarchical partitions of $[p]$. Define $\bell_0(p) = 1$ for all $p$.
\end{definition}
\begin{remark}
We note that higher-order Bell numbers (not to be confused with ordered Bell numbers) are also known as multidimensional Bell numbers, iterated Bell numbers, and iterated exponential integers.
\end{remark}
For example, $\bell_2(3) = 12$, since set $[3]=\{1,2,3\}$ has twelve 2-level partitions:
$$
\frac{\text{first level partition}}{\text{second level refinement}}:
\begin{array}{cccccc}
\frac{(1,2,3)}{(1)(2)(3)}&
\frac{(1,2,3)}{(1,2)(3)}&
\frac{(1,2,3)}{(1,3)(2)}&
\frac{(1,2,3)}{(2,3)(1)}&
\frac{(1,2,3)}{(1,2,3)}&
\frac{(1,2)(3)}{(1,2)(3)}\\[5mm]
\frac{(1,2)(3)}{(1)(2)(3)}&
\frac{(1,3)(2)}{(1,3)(2)}&
\frac{(1,3)(2)}{(1)(2)(3)}&
\frac{(2,3)(1)}{(2,3)(1)}&
\frac{(2,3)(1)}{(1)(2)(3)}&
\frac{(1)(2)(3)}{(1)(2)(3)}
\end{array}
$$
It will be convenient to use the following generalization of order-$k$ Bell numbers.
\begin{definition}
For a partition $T\in\calP$, let $\bell_k(T)$ be the number of $k$-level partitions of $p$ such that the first level partition is a refinement of $T$. Note that $\bell_k(p)= \bell_k(\{[p]\})$, where $\{[p]\}$ is the partition of $[p]$ into one part. Define $\bell_0(T)=1$ if $T = \{\{1\},\dots,\{p\}\}$, and $\bell_0(T) = 0$ otherwise.
\end{definition}

We will need the following observations.
\begin{itemize}
\item For all $k \geq 1$,
$$\bell_k(T) = \prod_{S\in T} \bell_k(|S|)$$
since each subdivision of $T$ can be obtained by independently partitioning each part $S\in T$.
\item For all $k \geq 1$,
$$\bell_k(T) = \sum_{T'\in \calP(T)}\bell_{k-1}(T').$$
Indeed, if $k\geq 2$, then the number of $k$-level partitions with the first level partition $T$ and second level partition $T'$ equals $\bell_{k-1}(T')$ when $T'\in \mathcal{P}(T)$. If $k=1$, then both sides are equal to 1.
\end{itemize}

\paragraph{Asymptotic formulas for higher-order Bell numbers}

Skau and Kristensen~\cite{skau2019asymptotic} showed that $\bell_k(p)\sim \frac{p!}{2^{p-1}}k^{p-1}$ when $p$ is fixed and $k \to \infty$. Therefore, for every fixed $p$ and all sufficiently large $k \geq \widehat k(p)$,  
\[
\bell_k(p)^{1/k} = \Theta(p k^{1-1/p}).
\]

Makarychev, Ovsiankin, and Tani established an upper bound on $\bell_k(p)^{1/p}$ that holds for all $p$ and~$k$ \cite{MOT24}. In Section~\ref{sec:bounds-on-bell}, we complement their result with an upper bound matching the lower bound up to a constant factor (slightly sharpening their upper bound in the regime where $k \approx \log^* p$). To state these results, we first recall the definitions of the iterated logarithm and the $\log^*$ function.

\begin{definition}
The iterated logarithm $\log_e^{(k)} p = \log_e\log_e\cdots\log_e p$ is the $k$-th iteration of the $\log_e(\cdot)$ function; $\log^{(0)} p = p$. We define $\log^*_e p$ as the smallest $k$ such that $\log_e^{(k)} p \leq 1$.
\end{definition}

The asymptotic growth of $\bell_k(p)^{1/p}$ is given by the following theorem.
 \begin{theorem}
    \label{thm:bkp_lower_bound}
    Let $k_0 = \log_e^* p - 1$.
    The following bounds on $\bell_k(p)$ hold.
    \begin{itemize}
        \item For $k \leq k_0$,
         \[
            \bell_k(p)^{1/p} = \Theta(p/\log_e^{(k)}(p)) .
         \]
         \item For $k > k_0$,
         \begin{align*}
            \bell_k(p)^{1/p} &= \Theta(p\cdot (k - k_0)^{1-1/p})\\
            &= \Theta(p\cdot \min(k^{1-1/p},k-k_0)).
         \end{align*}    
    \end{itemize}
 \end{theorem}


\subsection{Multilinear Products}\label{sec:multilinear}
First, we introduce a multilinear generalization of the standard inner product. As in the Euclidean case, we define both unscaled and scaled variants of the product and the $\ell_p$ norm to simplify notation. In each space, we will consistently use either the unscaled or the scaled variant, but not both.

\begin{definition}
\label{def:multilinear_product}
The unscaled multilinear product of vectors $v_1,\dots, v_p$ in $\mathbb{R}^{D}$ is given by
$$\langle v_1, \dots, v_p\rangle = \sum_{j=1}^D\prod_{i=1}^p v_i(j),$$
where $v_i(j)$ is the $j$-th coordinate of vector $v_i$.
We will also denote this product by  $\bigodot_{i=1}^p v_i$.
\end{definition}
Importantly, 
$$\|v\|_p^p = \langle v,\dots, v\rangle\qquad
\text{and}\qquad\bigodot_{i=1}^p (a_i \tensor b_i) = \bigodot_{i=1}^p a_i \cdot \bigodot_{i=1}^p b_i.$$
\begin{definition} We will associate a scaling factor $\alpha$ with some spaces $U$ and define the scaled multilinear product of vectors $v_1,\dots, v_p\in U$ as
$$
\langle v_1, \dots, v_p\Erangle = \alpha \sum_{j}\prod_{i=1}^p v_i(j).
$$
We will also denote this product by  $\Ebigodot_{i=1}^p v_i$.
Further, we will also consider the scaled $\ell_p$-norm defined as $\|v\|_\infty = \langle v,\dots,v\Erangle= \bigl(\alpha\sum_j v(j)^p\bigr)^{1/p}$.
However, we will never consider a scaled variant of the $\ell_{\infty}$-norm.
\end{definition}
\noindent  Given two  spaces $U_1$ and $U_2$ with scaling factors $\alpha_1$ and $\alpha_2$, we define the scaling factor in $U_1\tensor U_2$ to be $\alpha_1\alpha_2$ so that $\langle a_1\tensor b_1, \dots,a_p \tensor b_p\Erangle = \langle a_1,\dots,a_p\Erangle \cdot \langle b_1,\dots,b_p\Erangle$ for $a_i\in U_1$ and $b_j\in U_2$.

As in the case of $p=2$, we use the unscaled multilinear product in space $\mathbb{R}^{d_0}$ and a rescaled product in space $\mathbb{R}^{d}$ with $\alpha = 1/m$:
 $$\langle c(e_1),\dots, c(e_p)\Erangle = \frac{1}{m} \sum_{h\in \mathcal{E}} \langle c^h(e_1), \dots, c^h(e_p)\rangle =  {\mathbb{E}_h}\langle c^h(e_1),\dots, c^h(e_p)\rangle$$
 and accordingly $\|c(e)\|_p^p = {\mathbb{E}_h}[\|c^h(e)\|_p^p]$.


\section{Hardness Results for the General \texorpdfstring{$\ell_p$}{lp}-Shortest Path}\label{sec:analysis-for-general-p}
We now extend the results from Section~\ref{sec:reduction-l2} to all $\ell_p$-norms with integer $p\geq 2$, proving Theorem~\ref{thm:main-pvsnp}.
We use the same reduction as in Section~\ref{sec:reduction-l2} except that we define edge costs using vectors $v_i^c$ from an $(r, q, p)$-vector system $\calS$ rather than the $(r,q,2)$-vector system we used previously.

\subsection{The analysis of the reduction for an arbitrary \texorpdfstring{$p$}{p}}

We prove the following analogue of Theorem~\ref{thm:main}.
\begin{theorem} \label{thm:main_general} Fix $p 
\geq 2$ and consider an instance $\cal H$ of $r$-Hypergraph Label Cover and the corresponding instance $G$. 
\begin{itemize}
    \item If $\cal H$ is a yes-instance, then there is an $s$-$t$ path $P$ in $G^{\tensor k}$ with $\|\cost(P)\|_p^p = 1$. 
    \item  If $\cal H$ is a no-instance, then for every $s$-$t$ path $P$ in $G^{\tensor k}$,  $\|\cost(P)\|_p^p \geq \alpha_k \bell_k(p)$, where $\alpha_k = (1-p/r)^{pk} (1 - p^2\varepsilon)^k\geq 1 - O(p^2k(1/r + \varepsilon))$.
\end{itemize}
Additionally, the number of edges in the $\ell_p$-Shortest Path instance $|E(G^{\otimes k})| \leq |\mathcal{H}|^{2k}$. All cost vectors lie in a Euclidean space of dimension at most $r^{p \cdot k} \cdot |\mathcal{H}|^{2k}$.
\end{theorem}

For fixed integers $p \geq 2$ and $k$, we choose a sufficiently large $r$ and a sufficiently small $\varepsilon > 0$ so that the error term $O(p^2k(1/r + \varepsilon))$ in Theorem~\ref{thm:main_general} is smaller than the desired $\delta$. Then, from the hardness result for $r$-Hypergraph Label Cover (Theorem~\ref{thm:hardness-r-hypergraph-label-cover}), we immediately obtain that the $\ell_p$-Shortest Path problem does not admit a better than $((1 - \delta)\bell_k(p))^{1/p}$ approximation in series-parallel graphs of order $k$, yielding Theorem~\ref{thm:main-pvsnp}.
In the remainder of this section, we will prove Theorem~\ref{thm:main_general}.

\paragraph{Yes-case.} Assume first that $\mathcal{H}$ is a yes-instance. 
Then all hyperedges are satisfied by some assignment $\sigma$. 
Let $P = P_\sigma$. We show that $\|\cost(P)\|_p^p = 1$.
Consider a hyperedges $h = (u_1, \dots, u_r)$. Let $c$ be the color of all $u_i$ in $h$ w.r.t.\  $\sigma$. Then, by~\eqref{eq:vs-orthogonal} and \eqref{eq:vs-norm},
$$\Bigl\|\sum_{e\in P_{\sigma}} c^h(e)\Bigr\|_p^p =  \Bigl\|\sum_{j=1}^r v_{j}^{c}\Bigr\|^p_p = 
\sum_{j_1,\dots, j_p\in [r]} \langle v_{j_1}^{c},\dots, v_{j_p}^{c}\rangle = \sum_{j=1}^r \|v_j^c\|_p^p = 1.$$
We conclude that $\|\cost(P_\sigma)\|_p^p = 1$.
Now consider path $P^{\tensor k}$ in $G^{\tensor k}$. 
It consists of edges $\{(e_1,\dots, e_k):e_1,\dots, e_k \in P\}$.
Therefore, $\cost(P^{\tensor k}) = \cost(P)^{\tensor k}$. Since $\|\cost(P)\|_p^p = 1$, we get $\|\cost(P^{\tensor k})\|_p^p = 1$, as required.

\paragraph{No-case.} Now, assume that $\mathcal{H}$ is a no-instance. 
The following definition and lemma will complete the show the desired bound and complete the proof of Theorem~\ref{thm:main_general}.

\begin{definition}
Let us say that paths $P_1,\dots, P_p$ are consistent with a partition $T\in \calP$ if the following property holds: if $i$ and $j$ are in the same part of $T$, then $P_i = P_j$. (If $i$ and $j$ are in different parts, then $P_i$ and $P_j$ may or may not be equal.)
\end{definition}


\begin{lemma}\label{lem:main_general_nocase}
Consider $s$-$t$ paths $P_1,\dots, P_p$ in $G^{\tensor k}$. Assume that the paths are consistent with some $T\in \calP$.
Then
$$\Ebigodot_{i=1}^p \cost(P_i)\geq \alpha_k \bell_k(T).$$
\end{lemma}
Before we proceed with the proof of Lemma~\ref{lem:main_general_nocase}, note that by letting $P_1=\dots=P_k = P$ and $T=\{[p]\}$, we immediately get the second part of Theorem~\ref{thm:main_general}.

\begin{proof}[Proof of Lemma~\ref{lem:main_general_nocase}] We prove the statement by induction on $k$.

\paragraph{Base Case} The statement is trivial for $k=0$, since, on one hand, $\cost(P_i) = 1$ when $k=0$ and thus the expression on the left is 1; on the other hand, $\bell_0(T) \leq 1$.
\paragraph{Induction Step}
Consider paths $P_1, P_2,\dots, P_p$ and a partition $T$ as in the statement.
Let $P_i'$ be the projection of $P_i$ onto $G$ and $\sigma_i$ be the assignment defined by $P_i'$.
For an edge $e \in P_{i}'$, let $Q_i(e)$ be the path in $G^{\tensor(k-1)}$ corresponding to edge $e$. Then
$$\cost(P_i) = \sum_{e\in P_i'} c(e)\tensor \cost(Q_i(e)),$$
and therefore
\begin{align*}
\Ebigodot_{i=1}^p \cost(P_i) &= \sum_{e_1\in P_1',\dots, e_p\in P_p'}\Ebigodot_{i=1}^p c(e_i)\tensor \cost(Q_i(e_i)) \\
&=
\sum_{e_1\in P_1',\dots, e_p\in P_p'}\left(\Ebigodot_{i=1}^p c(e_i)\right)\cdot \left(\Ebigodot_{i=1}^p \cost(Q_i(e_i))\right)
\end{align*}
We are going to group tuples of $(e_1,\dots,e_p)$ in this summation based on the subdivision $T'$ of $T$ they define. Consider a choice of $e_1\in P_1',\dots, e_p\in P_p'$. Define a subdivision $T'$ of $T$ as follows: $i$ and $j$ are in the same part of $T'$ if they are in the same part of $T$ and $e_i=e_j$.
Note that if $i$ and $j$ are in the same part of $T'$, then  $P_i = P_j$ and $e_i = e_j$, and thus $Q_i(e_i) = Q_j(e_j)$. Therefore, paths $Q_1(e_1),\dots, Q_p(e_p)$ in $G^{\tensor(k-1)}$ are consistent with $T'$. By the induction hypothesis,
$$\Ebigodot_{i=1}^p \cost(Q_i(e_i)) \geq \alpha_{k-1}\bell_{k-1}(T').$$
We conclude that 
\begin{align*}
\Ebigodot_{i=1}^p \cost(P_i) &\geq \alpha_{k-1} \sum_{T'\in\calP(T)} \sum_{\substack{
e_1\in P_1',\dots,e_p\in P_p'\\\text{define }T'}}
\bell_{k-1}(T')\cdot 
\Ebigodot_{i=1}^p c(e_i)\\
&= \alpha_{k-1}\, \EE{h}{\sum_{T'\in\calP(T)} \bell_{k-1}(T')\cdot \sum_{\substack{
e_1\in P_1',\dots,e_p\in P_p'\\\text{define }T'}}\bigodot_{i=1}^p c^h(e_i)}.
\end{align*}

We now establish a lower bound on 
\begin{equation}\label{eq:sigma-T-prime}
\Sigma_{T'} \equiv\sum_{\substack{
e_1,\dots,e_p \\\text{ define }T'}}\bigodot_{i=1}^p c^h(e_i).
\end{equation}
Assume that $h$ is a colorful hyperedge w.r.t.\ $\sigma_1,\dots, \sigma_p$. 
Write $h=(u_1,\dots, u_r)$. Let $f_{ij}$ be the edge of $P_i'$ corresponding to vertex $u_j$. 
Further, let $c_{ij}$ be the color of $u_j$ in $h$ w.r.t.\  $\sigma_i$. Since $h$ is colorful, $c_{ij}\neq c_{i'j'}$ when $j\neq j'$.
Let us say that edges $f_{1j_1},\dots, f_{pj_p}$ with indices $j_1,\dots,j_p\in[r]$ satisfy Condition ($\star$) if:
\begin{quote}
\textbf{Condition ($\star$):} Edges $f_{1j_1},\dots, f_{pj_p}$ define subdivision $T'$ of $T$ and further $j_a\neq j_b$ if $a$ and $b$ lie in different parts of $T'$.
\end{quote}
We will now prove a lower bound on the contribution of $(f_{1j_1},\dots, f_{pj_p})$ satisfying Condition~($\star$) to $\Sigma_{T'}$. Since all terms in \eqref{eq:sigma-T-prime} are non-negative, this also provides a lower bound on $\Sigma_{T'}$ itself.

Observe that there are $|T'|$ distinct vectors among $\{v_{j_i}^{c_{ij_i}}\}_{i\in [p]}$ and, since $h$ is colorful, all of these distinct vectors have distinct colors. 
Therefore, by~\eqref{eq:vs-non-orthogonal},
$$\bigodot_{i=1}^p c^h(f_{ij_i}) = \bigodot_{i=1}^p v_{j_i}^{c_{ij_i}}\geq \frac{1}{r^{|T'|}}.$$
Finally, let us count the number of choices for edges $f_{1j_1},\dots, f_{pj_p}$ that satisfy Condition ($\star$) for a given $T'$. For every part $S$ of $T'$, we choose an index $j_S\in [r]$; for different parts, we choose different indices. Then we let $j_i = j_S$ if $i\in S$ and obtain $f_{1j_1},\dots, f_{pj_p}$. Clearly the obtained edges $f_{ij_i}$ define partition~$T'$.
Since we choose $|T'|$ distinct indices $j_S$, the total number of ways to choose $f_{1j_1},\dots, f_{pj_p}$ is $r\cdot (r-1) \cdots (r + 1 - |T'|) \geq (r - p)^{|T'|} \geq (1-p/r)^p \cdot r^{|T'|}$.
We conclude that
$$\sum_{\substack{
e_1,\dots,e_p\\\text{define }T'}}\bigodot_{i=1}^p c^h(e_i)
\geq (1-p/r)^p r^{|T'|} \cdot \frac{1}{r^{|T'|}} = (1 - p/r)^p$$
and
$$\sum_{T'\in\calP(T)} \bell_{k-1}(T')\cdot \sum_{\substack{
e_1,\dots,e_p\\\text{define }T'}}\bigodot_{i=1}^p c^h(e_i)
\geq (1-p/r)^p \sum_{T'\in\calP(T)} \bell_{k-1}(T')= (1 - p/r)^p \bell_k(T).$$
By Lemma~\ref{lem:colorful}, at least a $(1-p^2\varepsilon)$ fraction of all hyperedges are colorful w.r.t.\ $\sigma_1,\dots, \sigma_p$ and, therefore,
$$\Ebigodot_{i=1}^p \cost(P_i) \geq \alpha_{k-1}\cdot (1 - p^2 \varepsilon) \cdot (1-p/r)^p\cdot  \bell_k(T)= \alpha_k\bell_k(T).$$

\paragraph{Size Analysis}
The graph $G$ constructed in this reduction is the same as that in Section~\ref{sec:reduction-l2}, meaning that its size is $|E(G)| \leq |\mathcal{H}|^2$.
However, the dimension of the cost vectors assigned to the edges of $G$ differs compared to what we had in Section~\ref{sec:reduction-l2}, since we are using a different vector system. Using a similar argument to the one from the size analysis in Section~\ref{sec:reduction-l2}, we can upper bound the dimension of the cost vectors assigned to the edges of $G$ by $\max\{r^p,|\mathcal{H}| \cdot r\} \cdot |\mathcal{H}| \leq r^p \cdot |\mathcal{H}|^2$.
Thus, the instance of $\ell_p$-Shortest Path that we constructed in this section has size $|E(G^{\otimes k})| \leq |\mathcal{H}|^{2k}$ and uses cost vectors of dimension at most $r^{p \cdot k} \cdot |\mathcal{H}|^{2k}$.


\end{proof}

\section{Reduction with Super-constant \texorpdfstring{$k$}{k} and \texorpdfstring{$r$}{r}}
\label{sec:superconstant}
In this section, we prove Theorem~\ref{thm:hardness-of-ell-p-shortest-path}.
To this end, we show a reduction from 3SAT to $\ell_p$-Shortest Path by combining the tools we have developed in the previous sections of this paper and applying the Bell number lower bound from Theorem~\ref{thm:bkp_lower_bound}. 
The statement of our reduction is presented in the following claim:

\begin{claim}\label{cl : reduction from 3SAT to ell-p shortest path}
For some absolute constant $C \geq 2$, there exists an algorithm that receives a 3SAT formula of size $m \geq 2$, and integer parameters $p\geq 2$ and $N > 1$,
such that
\[\label{eq:lower-bound-on-N}
\log N \geq C \log m \cdot \max\{\log p,\log\log N\} \cdot \log^* p.
\]
The algorithm outputs an instance $\mathcal{I}$ of the $\ell_p$-Shortest Path problem satisfying the following conditions.
\begin{itemize}
\item \textbf{Yes-Instance.} If the input formula is satisfiable, then  $\mathcal{I}$ has a solution of cost at most $1$.
\item \textbf{No-Instance.} If the input formula is unsatisfiable, then every solution to instance $\mathcal{I}$ has cost at least
\begin{equation}\label{eq : ell-p cost of final solution}
    \frac{1}{C} \cdot p \cdot \left( \frac{\log N}{\log m \cdot \max\{\log p,\log\log N\}} \right)^{1-\frac{1}{p}} 
\end{equation}
\end{itemize}
Further, the graph in instance $\mathcal{I}$ is of size at most $N$ and the cost vectors have dimension at most $\ell = N^{1+\frac{p}{\log m}}$. The running time of the algorithm is $\mathrm{poly}(\ell)$.

\end{claim}

Before proving Claim~\ref{cl : reduction from 3SAT to ell-p shortest path}, we introduce some notation and recall previously discussed results.
The size of an instance $\mathcal{H}$ of $r$-Hypergraph Label Cover, denoted $|\mathcal{H}|$, is the maximum among the number of vertices in $\mathcal{H}$, the number of labels in $\mathcal{H}$, the number of colors in $\mathcal{H}$, and the number of hyperedges in $\mathcal{H}$.
We will use Theorem~\ref{thm:hardness-r-hypergraph-label-cover} with parameter $\varepsilon = \frac{1}{r}$. In that case, it provides a reduction that transforms 3SAT formulas of size $m$ to $r$-Hypergraph Label Cover  instances $\mathcal{H}$ of size $|\mathcal{H}| \leq m^{O(\log r)}$. We will denote the hidden constant in the big-$O$ notation by $c_1$. That is, the reduction produces instrances $\cal H$ of size at most $m^{c_1 \cdot \log r}$. Theorem~\ref{thm:bkp_lower_bound} guarantees that 
$\bell_k(p)^{1/p} \geq c_2^{-1} p(k - \log^* p)^{1-1/p}$ for some $c_2 \geq 1$. We will let $c_2$ be the smallest value greater than or equal to 1 for which this inequality holds.
We are now ready to prove Claim~\ref{cl : reduction from 3SAT to ell-p shortest path}.

\begin{proof}[Proof of Claim~\ref{cl : reduction from 3SAT to ell-p shortest path}]
Let $C = e \cdot 48 c_1 \cdot c_2$. We first describe the reduction.
Let 
\[
k' = \frac{\log N}{4c_1 \cdot \log m \cdot \max\{\log(2p^2),\log\log N\}}, \quad k = \lfloor k'\rfloor, \quad\text{and}\quad r = 2k \cdot p^2.
\]
The algorithm uses the reduction from Theorem~\ref{thm:hardness-r-hypergraph-label-cover} on the input formula $\varphi$ with $r =2kp^2$ (as defined above) and $\varepsilon = \frac{1}{r}$ (in the analysis, we will show that $k \geq 1$ and therefore this value of $k$ is valid).
Let $\mathcal{H}$ denote the resulting instance of $r$-Hypergraph Label Cover. Note that $\vert\mathcal{H}\vert \leq m^{c_1 \cdot \log r}$.
Next, the algorithm uses the reduction from $r$-Hypergraph Label Cover to $\ell_p$-Shortest Path from Theorem~\ref{thm:main_general} with $k$ as defined above.
The algorithm returns the resulting instance $\mathcal{I}$ of $\ell_p$-Shortest Path.

\paragraph{Bounding $k$}
We will begin the analysis by showing that 
\begin{equation}\label{eq : bounds on k}
1\leq 2\log^* p \leq k \leq \frac{\log N}{2c_1 \cdot \log m \cdot \log r}.
\end{equation}
Since $p\geq 2$, $\log^* p \geq 1$.
Using that $\log N \geq C \log m \cdot \max\{\log p,\log\log N\} \cdot \log^* p$ and our choice of $C$, we get that $k' \geq 3\log^* p$ and thus $k > k' - 1 \geq k' - \log^* p \geq 2\log^* p$.

We now prove the upper bound on $k$ stated in~\eqref{eq : bounds on k}.  
By the definition of $k$, we have $k \leq \log N$, and therefore $\log k \leq \log\log N$. This implies that  
\[
\log(2p^2) + \log\log N \geq \log(2p^2) + \log k = \log r.
\]  
It follows that  
\[
\max\{\log(2p^2), \log\log N\} \geq \frac{\log r}{2}.
\]  
From the definition of $k'$, we obtain  
\[
k \leq k' \leq \frac{\log N}{4c_1 \cdot \log m \cdot (\log r)/2}.
\]
This concludes the proof of inequality~\eqref{eq : bounds on k}.

\paragraph{Correctness Analysis}
We prove that the optimal solution cost of instance $\mathcal{I}$ is bounded as described in the statement of the claim.
First, if the input 3SAT formula is satisfiable, then Theorem~\ref{thm:hardness-r-hypergraph-label-cover} guarantees that the instance $\mathcal{H}$ of $r$-Hypergraph Label Cover is completely satisfiable.
In this case, Theorem~\ref{thm:main_general} tells us that the instance $\mathcal{I}$ has a solution with cost $1$, as required.

Now we assume that the input 3SAT formula is unsatisfiable and show that the cost of every solution to instance $\mathcal{I}$ is at least~\eqref{eq : ell-p cost of final solution}.
Theorem~\ref{thm:hardness-r-hypergraph-label-cover} guarantees that every assignment for the instance $\mathcal{H}$ of $r$-Hypergraph Label Cover weakly satisfies at most a $\varepsilon = \frac{1}{r}$ fraction of the hyperedges.
Then, Theorem~\ref{thm:main_general} guarantees that every solution to instance $\mathcal{I}$ has an $\ell_p$-cost of at least $(\alpha_k \bell_k(p))^{1/p}$ where $\alpha_k = (1-p/r)^{pk} (1 - p^2\varepsilon)^k$.
Since $\varepsilon = \frac{1}{r}$ and $r = 2kp^2$, we get that $\alpha_k \geq e^{-2}$, so every solution to instance $\mathcal{I}$ has an $\ell_p$-cost of at least $e^{-2/p} \cdot \bell_k(p)^{1/p} \geq e^{-1} \cdot \bell_k(p)^{1/p}$.
Furthermore, since $k \geq 2\log^* p$ (see~\eqref{eq : bounds on k}), by Theorem~\ref{thm:bkp_lower_bound}, we have
\begin{align*}
    \bell_k(p)^{1/p}
    \geq \frac{1}{c_2}p \cdot (k - \log^* p)^{1-1/p}
    \geq \frac{1}{c_2}p \cdot(k/2)^{1-1/p}
    \geq \frac{1}{2c_2}p \cdot k^{1-1/p}.
\end{align*}
Thus, the $\ell_p$-cost of every solution to instance $\mathcal{I}$ is at least $p \cdot k^{1-1/p}/({2e \cdot c_2})$.
Now observe that
\begin{align*} 
    k
    \geq \frac{k'}{2}
    = \frac{\log N}{8c_1 \cdot \log m \cdot \max\{\log(2p^2),\log\log N\}}
    \geq \frac{\log N}{24c_1 \cdot \log m \cdot \max\{\log p,\log\log N\}}.
\end{align*}
This means that the $\ell_p$-cost of every solution to instance $\mathcal{I}$ is at least
\[
    \frac{p}{2e \cdot c_2} \cdot \left(\frac{\log N}{24c_1 \cdot \log m \cdot \max\{\log p,\log\log N\}}\right)^{1-1/p}
    \geq \frac{p}{C} \cdot \left(\frac{\log N}{\log m \cdot \max\{\log p,\log\log N\}}\right)^{1-1/p}.
\]
Define $\ell = N^{1+\frac{p}{\log m}}$ (as in the statement of the claim).
We prove that 
$$\left(\frac{\log N}{\log m \cdot \max\{\log p,\log\log N\}}\right)^{1-1/p} \geq \left( \frac{\log \ell}{(p+\log m) \cdot \log\log \ell} \right)^{1-\frac{1}{p}}.$$
 From the definition of $\ell$, we derive $\log m \cdot \log \ell = (p + \log m) \cdot \log N$, and thus
\begin{equation}\label{eq : intermediate equation for proving inequality relating N and ell}
    \left(\frac{\log N}{\log m \cdot \max\{\log p,\log\log N\}}\right)^{1-1/p} = \left( \frac{\log \ell}{(p+\log m) \cdot \max\{\log p,\log\log N\}} \right)^{1-\frac{1}{p}} \geq
\end{equation}
$$
\frac{p}{C}\cdot \left( \frac{\log \ell}{(p+\log m) \cdot \log\log \ell} \right)^{1-\frac{1}{p}}.
$$
In the last inequality, we used the fact that $\ell \geq 2^p$, which follows from the definition of $\ell = N^{1+\frac{p}{\log m}} = 2^{\log N +\frac{p \log N}{\log m}}$ as $\log N \geq \log m$.

\paragraph{Size Analysis}
We now upper bound the size of the resulting instance $\mathcal{I}$ of $\ell_p$-Shortest Path.
Recall that Theorem~\ref{thm:main_general} guarantees that the number of edges in the output graph is at most $\vert\mathcal{H}\vert^{2k}$ and the cost vectors have at most $r^{p \cdot k} \cdot \vert\mathcal{H}\vert^{2k}$ coordinates.
Further, $\vert\mathcal{H}\vert \leq m^{c_1 \log r}$ holds by our choice $c_1$. 
Thus, the size of the graph is at most
\begin{equation}\label{eq : bound on size of I in full reduction}
    (m^{c_1 \log r})^{2k}
    = 2^{2c_1 \cdot k \cdot \log r \cdot \log m}
    \leq 2^{\log N}
    = N
\end{equation}
here we applied the upper bound on $k$ from~\eqref{eq : bounds on k}. The dimension of the cost vectors is at most
\begin{align*}
    r^{p \cdot k} \cdot (m^{c_1 \log r})^{2k}
    = 2^{p \cdot k \cdot \log r + 2k \cdot \log r \cdot c_1 \log m}
    \leq 2^{2k \cdot \log r \cdot c_1\log m \cdot (1+\frac{p}{\log m})}
    \leq N^{1+\frac{p}{\log m}}
    =\ell.
\end{align*}

\paragraph{Running Time Analysis}
The reduction in Theorem~\ref{thm:hardness-r-hypergraph-label-cover} runs in time $m^{O(\log \frac{r}{\varepsilon})} = m^{O(\log r)}$. Theorem~\ref{thm:main_general} provides a polynomial-time reduction, with running time $\mathrm{poly}(N \cdot \ell) = \mathrm{poly}(\ell)$.
Since $\ell \geq |\mathcal{H}| = m^{O(\log r)}$, the $\mathrm{poly}(\ell)$ term dominates the overall running time.  
Thus, the total running time of the algorithm is $\mathrm{poly}(\ell)$.

\end{proof}

\subsection{Proof of Theorem~\ref{thm:hardness-of-ell-p-shortest-path}}\label{sec : proof of hardness of ell-p shortest path}

In this section, we prove Theorem~\ref{thm:hardness-of-ell-p-shortest-path}.

\begin{proof}[Proof of Theorem \ref{thm:hardness-of-ell-p-shortest-path}]
Assume that $NP \nsubseteq \bigcap_{\delta > 0} \mathrm{DTIME}(2^{n^{\delta}})$. Let $c \geq 1$ and $p:\bbN \to \bbN$, $\alpha:\bbN \to \bbN$ be as in the statement of Theorem \ref{thm:hardness-of-ell-p-shortest-path}.
Since $NP \nsubseteq \bigcap_{\delta > 0} \mathrm{DTIME}(2^{n^{\delta}})$, $\mathrm{3SAT} \notin \bigcap_{\delta > 0} \mathrm{DTIME}(2^{n^{\delta}})$.
Fix $\delta^* \in (0, 1/2)$ such that $\mathrm{3SAT} \notin \mathrm{DTIME}(2^{n^{\delta^*}})$.
Define $C' = 2\cdot C/\delta^*$ where $C \geq 1$ is the constant from Claim \ref{cl : reduction from 3SAT to ell-p shortest path}.

We prove that there is no algorithm for $\ell_p$-Shortest Path with the following guarantee: given an instance of $\ell_p$-Shortest Path of size $N$, with  $p=p(N)$, the algorithm returns an $\alpha(N)$-approximate solution in time $O(2^{(\log N)^c})$.
Assume towards contradiction that there exists an algorithm $\mathcal{A}$  with the above guarantee.
We will show that under this assumption, $3SAT \in \mathrm{DTIME}(2^{n^{\delta^*}})$.
Consider the following algorithm $\cal B$ for solving a $3SAT$ instance $\varphi$ of size $m$:
Run the algorithm from Claim~\ref{cl : reduction from 3SAT to ell-p shortest path} on $\varphi$ with parameters $p=p(\lfloor 2^{m^{\delta^*/c}}\rfloor)$ and $N=\lfloor 2^{m^{\delta^*/c}/(1+\frac{p}{\log m})}\rfloor$.  
Take the resulting $\ell_p$-Shortest Path instance $\mathcal{I}$ and run $\mathcal{A}$ on it to get the final result.

We will assume below that $m$ and $n$ are sufficiently large. 
We need to show that the parameters with which $\cal B$ calls the subroutine from Claim~\ref{cl : reduction from 3SAT to ell-p shortest path} satisfy the condition $\log N \geq C \log m \cdot \max\{\log p,\log\log N\} \cdot \log^* p$ in the statement of that claim.
Additionally, we need to show the algorithm $\mathcal{A}$ can differentiate between the yes- or a no-cases, as described in the statement of Claim~\ref{cl : reduction from 3SAT to ell-p shortest path}.
\begin{itemize}
    \item We begin by proving that
    \begin{equation}\label{eq : condition from full reduction that we need to prove}
        \log N \geq C \log m \cdot \max\{\log p,\log\log N\} \cdot \log^* p.
    \end{equation}
    By our choice of the parameters $N$ and $p$, both the graph size and the dimension of the cost vectors in instance $\mathcal{I}$ are upper bounded by $\ell=\lfloor2^{m^{\delta^*/c}}\rfloor$. By our assumption, $p = p(\ell)$  satisfies $p(\ell) \cdot \log^* p(\ell) \leq \frac{\log \ell}{C' \cdot c \cdot \log \log \ell}$ for all large enough $\ell$. Since $\ell$ is strictly increasing as a function of $m$, we get that this holds for all sufficiently large $m$.
    Furthermore, since $\ell = \lfloor 2^{m^{\delta^*/c}}\rfloor$,
    we get that $\log m \cdot \log^* p \leq \frac{\log \ell}{C'\cdot c\cdot \log\log \ell}$ for sufficiently large $m$.
    Thus,
    \begin{equation}\label{eq : intermidiate condition relating ell and p}
        \frac{2\log \ell}{C'\cdot c \cdot \log \log \ell} \geq (p+\log m) \cdot \log^* p.
    \end{equation}
    In particular, we have $\log \ell \geq p$ and thus, $\log\log \ell \geq \max\{\log (p),\log\log N\}$.
    Plugging this last inequality into~\eqref{eq : intermidiate condition relating ell and p}, we get 
    \[
     \frac{2}{C'\cdot c}\log \ell \geq (p+\log m) \cdot \max\{\log (p),\log\log N\} \cdot \log^* p.
    \]
    Finally, since $\log \ell \leq \frac{p+\log m}{\log m}\log N$ and $C' \geq 2\cdot C$, the inequality above implies~\eqref{eq : condition from full reduction that we need to prove}.

    \item We will now prove that the algorithm $\mathcal{A}$ can differentiate between yes- and no-instance.
    To this end, we need to show that the following inequality holds
    \begin{equation}\label{eq : approximation is less than gap}
        \alpha(\ell) 
        < \frac{1}{C} \cdot p \cdot \left(\frac{\log \ell}{(p+\log m)\log\log \ell}\right)^{1-\frac{1}{p}}.
    \end{equation}
    Since $p=p(\ell)$ and $\ell$ is a strictly increasing function of $m$, we have 
    \[
        \min\left\{p \cdot \left(\frac{\log \ell}{\log^2\log \ell}\right)^{1-\frac{1}{p}}, \frac{\log \ell}{\log \log \ell} \right\} 
        < 
        p\cdot\left(\frac{\log \ell}{(p + \log \log \ell)\log\log \ell}\right)^{1-\frac{1}{p}}
    \]
    and thus, 
    \[
    \alpha(\ell) < \frac{1}{C' \cdot c} \cdot         p\cdot\left(\frac{\log \ell}{(p + \log \log \ell)\log\log \ell}\right)^{1-\frac{1}{p}}.
    \]
    Using that $\log\log \ell \geq \log\log 2^{m^{\delta^*/c}} - 1= \delta^*/c \cdot \log m - 1$, we obtain
    \begin{align*}
        \alpha(\ell)
        &\leq \frac{1}{C' \cdot c} \cdot p\cdot\left(\frac{\log \ell}{(p + \delta^*/c \cdot \log m - 1)\log\log \ell}\right)^{1-\frac{1}{p}}\\
        &\leq \frac{1}{C' \cdot c} \cdot p\cdot\left(\left(\frac{c}{\delta^*}\right)\frac{\log \ell}{(p +  \log m)\log\log \ell}\right)^{1-\frac{1}{p}}\\
        &\leq \frac{1}{C' \cdot c} \cdot p\cdot \left(\frac{c}{\delta^*}\right) \cdot \left(\frac{\log \ell}{(p +  \log m)\log\log \ell}\right)^{1-\frac{1}{p}} 
    \end{align*}
    where the last two inequalities use that $c/\delta^* \geq 2$.
    Since $C' > C/\delta^*$, the last equation implies~\eqref{eq : approximation is less than gap}.
    This completes the proof that~\eqref{eq : approximation is less than gap} holds for sufficiently large $m$, implying that $\mathcal{A}$ can indeed determine whether instance $\mathcal{I}$ is a yes-instance or a no-instance for sufficiently large $m$.
\end{itemize}

It remains to upper bound the running time of the algorithm.
The algorithm from Claim~\ref{cl : reduction from 3SAT to ell-p shortest path} runs in time polynomial in $\ell$.
The instance $\mathcal{I}$ of $\ell_p$-Shortest Path produced has both the graph size and the dimension of the cost vectors upper bounded by $\ell$. 
Thus, the running time of the algorithm $\mathcal{A}$ on instance $\mathcal{I}$ is upper bounded by $O(2^{(\log \ell)^c}) = O(2^{m^{\delta^*}})$. 
It follows that the running time of $\mathcal{B}$ is $O(2^{m^{\delta^*}})$ which contradicts $\mathrm{3SAT} \notin \mathrm{DTIME}(2^{n^{\delta^*}})$. 
\end{proof}

\section{Hardness of \texorpdfstring{$\ell_{\infty}$}{ℓ∞}-Shortest Path}\label{sec:infinity}
In this section, we consider the case when $p = \infty$ and prove Theorem~\ref{thm: hardness of ell-infinity shortest path}. 
We begin with the following claim, which  captures the details of our reduction. 

\begin{claim}\label{clm : reduction for ell-infinity}
There exists an algorithm that receives a 3SAT formula $\varphi$ of size $m$ and integer parameters $k \geq 2$ and outputs an instance $\mathcal{I}$ of the $\ell_\infty$-Shortest Path problem satisfying the following conditions:
\begin{itemize}
\item \textbf{Yes-Instance.} If the input formula is satisfiable, then $\mathcal{I}$ has a solution of cost at most $1$.
\item \textbf{No-Instance.} If the input formula is unsatisfiable, then every solution to instance $\mathcal{I}$ has cost at least $\frac{k^2}{900}$.
\end{itemize}
Further, the size of the graph in instance $\mathcal{I}$ and the cost vectors' dimension are both at most $2^{O(k \log k \log m)}$. The running time of the algorithm is $m^{\operatorname{poly}(k)}$.
\end{claim}

Throughout this section, we assume that $NP \nsubseteq \bigcap_{\delta > 0} \mathrm{DTIME}(2^{n^{\delta}})$.
The proof of Claim~\ref{clm : reduction for ell-infinity} is presented in Section~\ref{sec : reduction for ell-infinity} and requires the notion of a \emph{modified vector system} that we describe in Section~\ref{sec: modified vector system}.

\subsection{Proof of Theorem~\ref{thm: hardness of ell-infinity shortest path}}

In this section, we prove Theorem~\ref{thm: hardness of ell-infinity shortest path} assuming Claim~\ref{clm : reduction for ell-infinity}, which is proved later.

We say that a deterministic approximation algorithm $\mathcal{A}$ for the $\ell_{\infty}$-Shortest Path problem is \emph{$c$-good} if it finds an $\left(\frac{\log n}{C'' \cdot c \cdot (\log \log n)^2}\right)^2$-approximate solution in time $O(2^{(\log n)^c})$, where $n$ is the size of the input instance.
Our goal is to prove that there is no $c$-good approximation algorithm for every $c\geq 1$.

Our complexity assumption $NP \nsubseteq \bigcap_{\delta > 0} \mathrm{DTIME}(2^{n^{\delta}})$ implies that $\mathrm{3SAT} \notin \bigcap_{\delta > 0} \mathrm{DTIME}(2^{n^{\delta}})$. Thus, $\mathrm{3SAT} \notin \mathrm{DTIME}(2^{n^{\delta^*}})$ for some $\delta^*\in(0,1)$. 

Our proof of Theorem~\ref{thm: hardness of ell-infinity shortest path} is by contradiction. We assume that there is a $c$-good approximation algorithm $\cal A$ for some $c \geq 1$ and then show how to use $\cal A$ to solve $\mathrm{3SAT}$ in time $O(2^{n^{\delta^*}})$.

Claim~\ref{clm : reduction for ell-infinity} presents a reduction that converts a 3SAT formula of size $m$ to an $\ell_{\infty}$-Shortest Path instance of size at most $2^{c^* \cdot k \log k \log m}$, in time $O\left(m^{k^{c^*}}\right)$, for some $c^* \geq 1$. We fix such a constant $c^* \geq 1$ for the remainder of the proof.

We now present an algorithm that solves $\mathrm{3SAT}$ in time $O(2^{n^{\delta^*}})$. The algorithm receives a 3SAT formula $\varphi$ of size $m$.  Define 
\[
k=\left\lfloor m^{\delta^*/c'}/(C^*\log m)^2\right\rfloor,
\]
where $c'$ is given by $c' = c \cdot c^* \geq 1$.
We assume that $m$ is sufficiently large and, in particular, 
$k > 30$, $k > \sqrt{m^{\delta^*/c'}}$, and $m\geq 2$ (otherwise, we solve $\varphi$ via brute-force search). Let
\[
k' = m^{\delta^*/c^*}/(\log m)^{1/c^*} \geq m^{\delta^*/c'}/(c^*\log m)^2 \geq k.
\]

We use the reduction from Claim~\ref{clm : reduction for ell-infinity} on $\varphi$ with parameter $k$ and obtain an instance $\mathcal{I}$.
The reduction running time is upper bounded by $$O\left(m^{k^{c^*}}\right) \leq O\left(m^{(k')^{c^*}}\right) = O\left(m^{\left(m^{\delta^*}/\log m\right)}\right) = O\left(2^{m^{\delta^*}}\right).$$
Further, the size of $\mathcal{I}$ is $|\mathcal{I}|\leq 2^{c^* \cdot k \log k \log m}$.
Also, since $k \geq \sqrt{m^{\delta^*/c'}} = m^{\delta^*/(2c')}$,  we have $\log m \leq \frac{2c'}{\delta^*}\log k = \frac{2c^* \cdot c}{\delta^*}\log k$ and thus 
$$|\mathcal{I}|\leq 2^{\left(\frac{2(c^*)^2 \cdot c}{\delta^*} \cdot k (\log k)^2\right)}.$$
We can also upper bound $|\mathcal{I}|$ by a function of $m$ using the upper bound $|\mathcal{I}|\leq 2^{c^* \cdot k \log k \log m}$ as follows. Using that $m \geq 2$, we get
$$k=\left\lfloor m^{\delta^*/c'}/(c^*\log m)^2\right\rfloor \leq m^{\delta^*/c'}/(c^*\log m)^2.$$
Also, $\log k \leq \log m$. Therefore, 
$$k \log k \log m \leq k (\log m)^2 \leq \frac{1}{(c^*)^2} m^{\delta^*/c'} \leq \frac{1}{(c^*)^2} m^{\delta^*/c},$$
where the last inequality used that $c'=c \cdot c^* \geq c$. Now the inequality $|\mathcal{I}|\leq 2^{c^* \cdot k \log k \log m}$ implies that $|\mathcal{I}|\leq 2^{\left(\frac{1}{c^*} \cdot m^{\delta^*/c}\right)} \leq 2^{\left(m^{\delta^*/c}\right)}$.
To summarize the algorithm so far, we spent $O\left(2^{m^{\delta^*}}\right)$ time to construct an instance $\mathcal{I}$ of the $\ell_{\infty}$-Shortest Path problem, whose size is upper bounded both by $2^{\left(\frac{2(c^*)^2 \cdot c}{\delta^*} \cdot k (\log k)^2\right)}$ and by $2^{\left(m^{\delta^*/c}\right)}$. By Claim~\ref{clm : reduction for ell-infinity}: if $\varphi$ is satisfiable then $\mathcal{I}$ has a solution of value at most $1$; and, if $\varphi$ is unsatisfiable, then every solution to $\mathcal{I}$ has value at least $\frac{k^2}{900}$.
It now remains to distinguish between the case where the optimal solution value of $\mathcal{I}$ is $\operatorname{OPT} \leq 1$, and the case where it is $\operatorname{OPT} \geq \frac{k^2}{900}$.

Next, we run the algorithm $\mathcal{A}$ on $\mathcal{I}$.
Since the algorithm $\mathcal{A}$ is $c$-good and since the size of $\mathcal{I}$ is upper bounded by $2^{\left(m^{\delta^*/c}\right)}$, the running time of this call is upper bounded by $O\left(2^{(\log |\mathcal{I}|)^c}\right) = O\left(2^{m^{\delta^*}}\right)$.
Furthermore, since $\mathcal{A}$ is $c$-good, the solution returned by $\mathcal{A}$ must be $\alpha$-approximate optimal solution of $\mathcal{I}$, for $\alpha =\max\left\{1,\left(\frac{\log |\mathcal{I}|}{C'' \cdot c \cdot (\log \log |\mathcal{I}|)^2}\right)^2\right\}$.
We now upper bound $\alpha$ by a function of $k$: define $h = \frac{2(c^*)^2 \cdot c}{\delta^*} \cdot k (\log k)^2$; since the size of $|\mathcal{I}|$ is upper bounded by $2^{h}$, it follows that $\frac{\log |\mathcal{I}|}{(\log \log |\mathcal{I}|)^2} \leq \frac{16h}{(\log h)^2} \leq \frac{16h}{(\log k)^2} = \frac{32(c^*)^2 \cdot c \cdot k}{\delta^*}$; so, $\alpha \leq \max\left\{1,\left(\frac{32(c^*)^2 \cdot k}{C'' \cdot \delta^*}\right)^2\right\}$.
We now reveal the value of the absolute constant $C''$ from the statement of Theorem~\ref{thm: hardness of ell-infinity shortest path} to be $C''=\frac{961(c^*)^2}{\delta^*}$, meaning that $\alpha \leq \max\left\{1,\left(\frac{960}{961}\right)^2 \cdot \left(\frac{k}{30}\right)^2\right\}$.
Thus, from the assumption  $k > 30$, we get $\alpha < \left(\frac{k}{30}\right)^2 = \frac{k^2}{900}$.
Therefore, the value of the $\alpha$-approximate solution produced by $\mathcal{A}$ can be used to distinguish between the case where the optimal solution value of $\mathcal{I}$ is $\operatorname{OPT} \leq 1$, and the case where it is $\operatorname{OPT} \geq \frac{k^2}{900}$, as required.
In the former case, we report that the 3SAT formula $\varphi$ is satisfiable; in the latter case, we report that it is unsatisfiable.
This concludes the description of the algorithm for solving a 3SAT formula $\varphi$ of size $m$.
Note that the running time of each part of the algorithm is at most $O\left(2^{m^{\delta^*}}\right)$, as required.

To complete the proof of Theorem~\ref{thm: hardness of ell-infinity shortest path}, we now prove Claim~\ref{clm : reduction for ell-infinity}.

\subsection{Modified Vector Systems}\label{sec: modified vector system}

In this subsection, we introduce the notion of a modified vector system, which is conceptually similar to the vector system defined in Definition~\ref{def : vector system}, but provides slightly different guarantees.

\begin{definition}\label{def : modified vector system}
Consider a collection of vectors $\calS = \{v_i^c: i\in [r], c\in [q]\}$ (all vectors $v_i^c$ are distinct) whose entries are in $\{0,1\}$. Let us say that $c$ is the color of vector $v_i^c$.
$\calS$ is a $(r, q, p)$-modified vector system if the following conditions hold:
\begin{itemize}
    \item  For every two distinct vectors $u_1 \neq u_2$ of the same same color, we have 
    $$\left\|u_1 + u_2\right\|_\infty \leq 1.$$
    \item For every $p' \in [p]$ and a sequence of $p'$ vectors $u_1, \dots, u_{p'}\in {\cal S}$ (the same vector may appear more than once in this sequence), in which every two distinct vectors $u_i\neq u_j$ have different colors, we have
    $$\Biggl\|\sum_{i=1}^{p'} u_i\Biggr\|_\infty \geq \frac{p'}{r} + \frac{1}{10}\sqrt{\frac{p'}{r}}.$$
\end{itemize}
\end{definition}

The next claim shows that for $r = 2$, there exist $(r, q, p)$-modified vector systems with vectors of dimension $O(p \log q)$ and provides an algorithm for constructing such systems.


\begin{claim}\label{clm: constructing modified vector systems}
There exists a deterministic algorithm that, given integers $p, q \geq 2$, constructs a $(2, q, p)$-modified vector system $\mathcal{S}$ with vectors of dimension $O(p \log q)$, in time $2^{O(pq \log q)}$.
\end{claim}

Before proving Claim~\ref{clm: constructing modified vector systems}, we state the following simple claim, whose proof appears in Appendix~\ref{sec:Bernoulli}.

\begin{claim}\label{cl: claim about Bernouli random variables}
Let $Y_1, \ldots, Y_k \in \{0,1\}$ be independent, unbiased Bernoulli random variables, and let $a_1, \ldots, a_k$ be non-negative real numbers. Define the random variable
$
X = \sum_{i=1}^k a_i Y_i,
$
and let $\mu = \mathbb{E}[X]$ and $\sigma = \sqrt{\mathrm{Var}(X)}$ denote its mean and standard deviation, respectively. Then, for every $c \in [0,1]$,
\[
\Pr\left(X \geq \mu + c\sigma\right) \geq \frac{(1 - c)^2}{4}.
\]
\end{claim}
We now present the proof of Claim~\ref{clm: constructing modified vector systems}.

\begin{proof}[Proof of Claim~\ref{clm: constructing modified vector systems}]
Fix $d_0=\roundup{100 \cdot p \log q} = O(p \log q)$.
We first give a simple brute-force search algorithm that either finds a $(2,q,p)$-modified vector system where the vectors have dimension $d_0$, or correctly declares that no such system exists.
Next, we prove that a $(2,q,p)$-modified vector system of dimension $d_0$ always exists.

\paragraph{The Algorithm}
Observe that a $(2,q,p)$-modified vector system contains exactly $2q$ vectors, and all of the entries of these vectors belong to $\{0,1\}$.
To find such a modified vector system, we iterate over all sequences $(v_1,\ldots,v_{2q})$ of $2q$ vectors from $\{0,1\}^{d_0}$, and check for each such sequence whether it constitutes a $(2,q,p)$-modified vector system.
Performing this check for a single sequence $(v_1,\ldots,v_{2q})$ requires going over all positive integers $p' \leq p$ and all possible multisets of cardinality $p'$ containing only vectors from the sequence $(v_1,\ldots,v_{2q})$, and spending $\operatorname{poly}(p' \cdot qd_0)=\operatorname{poly}(pq)$ time on each such multiset.
As there are at most $\sum_{p'=1}^{p} (2q)^{p'} \leq p (2q)^{p}$ such multisets for each sequence $(v_1,\ldots,v_{2q})$, and as there are at most $2^{2qd_0}=2^{O(qp\log q)}$ such sequences, we get that the total running time of the algorithm is $2^{O(qp \log q)} \cdot p(2q)^{p} \cdot \operatorname{poly}(pq) = 2^{O(qp \log q)}$, as required.

\paragraph{Existence of a Vector System}
It remains to prove that there exists a $(2,q,p)$-modified vector system whose vectors are all of dimension $d_0$.
To prove this, we present a randomized construction that obtains a $(2,q,p)$-modified vector system $\calS = \{v_b^c: b\in [2], c\in [q]\}$ of dimension $d_0$ with non-zero probability.
In this construction, for every color $c \in [q]$, we treat the two vectors $v^c_1,v^c_2 \in \{0,1\}^{d_0}$ as representing two subsets of $[d_0]$, and we construct them so that these two subsets form a partition of the set $[d_0]$.
Before describing how we select these partitions, we observe that this construction will satisfy the first bullet point in Definition~\ref{def : modified vector system}, because, among two distinct vectors $u_1$ and $u_2$ of the same color $c \in [q]$, 
one will be equal to $v^c_1$ and the other to $v^c_2$. Accordingly, $\|u_1+u_2\|_\infty = \|v^c_1 + v^c_2\|_\infty = 1$.

We now describe how we select the partitions. To this end, for every color $c \in [q]$ and every element $j \in [d_0]$, we choose independently and uniformly at random whether to add $j$ to the set represented by $v^c_1$ or to the set represented by $v^c_2$.

It remains to show that, with non-zero probability, this system satisfies the second bullet point in Definition~\ref{def : modified vector system}, and all the vectors in this system are distinct.
Let $\mathcal{E}$ denote the bad event that this system does not satisfy the second bullet point in Definition~\ref{def : modified vector system}.
Let $\mathcal{E}'$ denote the bad event that there exist two distinct colors $c,c' \in [q]$ and two indices $b,b' \in \{1,2\}$ (where $b$ and $b'$ may be the same), such that $v^{c}_{b}=v^{c'}_{b'}$.
If $\mathcal{E}'$ does not occur, then all of the vectors in the system are distinct, becuase two vectors of the same color are always distinct.

It now remains to show that, with non-zero probability, neither of the bad events $\mathcal{E}$ and $\mathcal{E}'$ happens.
Specifically, it suffices to show that 
$\Pr[\mathcal{E}] + \Pr[\mathcal{E}'] < 1$.

We begin by proving that $\Pr[\mathcal{E}']<\frac{1}{2}$: observe that, for each color $c \in [q]$ and for each $b \in \{1,2\}$, the vector $v^c_b$ is uniformly distributed over $\{0,1\}^{d_0}$.
Furthermore, vectors of different colors are pairwise independent. Hence, for every two distinct colors $c,c' \in [q]$ and for every two indices $b,b' \in \{1,2\}$, we have $\Pr[v^{c}_{b}=v^{c'}_{b'}] = 2^{-d_0}$. 
By applying a union bound over all $4\binom{q}{2}$ possible choices of $c, c' \in [q]$ and $b, b' \in \{1,2\}$, we obtain
\[
\Pr[\mathcal{E}'] \leq 4\binom{q}{2} \cdot 2^{-d_0} \leq 4q^2 \cdot 2^{-d_0} \leq 4q^2 \cdot 2^{-6\log q} < \frac{1}{2}.
\]

Finally, we prove that $\Pr[\mathcal{E}] \leq \frac{1}{2}$.
For this proof, we use the following definition: 

\begin{definition}
    A sequence $P$ is a \emph{valid sequence of pairs} if every element of $P$ is a pair $(c,b) \in [q] \times \{1,2\}$, and every two elements $(c_1,b_1),(c_2,b_2)$ of $P$ are either equal, or have $c_1 \neq c_2$.
\end{definition}

Observe that the bad event $\mathcal{E}$ happens if and only if there exists an integer $p' \in [p]$ and a valid sequence of pairs $P=((c_1,b_1),\ldots,(c_{p'},b_{p'}))$ of length $p'$ which satisfy the inequality
\begin{equation}\label{eq: opposite of condition 2 of modified vector system}
    \left\|\sum_{i=1}^{p'} v^{c_i}_{b_i}\right\|_\infty < \frac{p'}{2} + \sqrt{\frac{p'}{200}}.
\end{equation}
Furthermore, observe that there are at most $\sum_{p'=1}^{p}(2q)^{p'} \leq p \cdot (2q)^p$ different possible ways of choosing an integer $p' \in [p]$ and a valid sequence of pairs $P$ of length $p'$.
Therefore, the desired inequality $\Pr[\mathcal{E}] < \frac{1}{2}$ follows by a union bound argument combined with the following claim.

\begin{claim}\label{cl: claim in construction of modified vector system}
    Consider any integer $p' \in [p]$, and any valid sequence of pairs $P=((c_1,b_1),\ldots,(c_{p'},b_{p'}))$ of length $p'$.
    Then, the probability that~\eqref{eq: opposite of condition 2 of modified vector system} holds with respect to $p'$ and $P$, is at most $\frac{1}{2 \cdot p \cdot (2q)^p}$.
\end{claim}
\begin{proof}
    Consider $p' \in [p]$ and a valid sequence of pairs $P$ of length $p'$, as in the statement of the claim.
    Let $p''$ be the number of unique pairs in $P$, and let $P'$ be the sequence obtained from $P$ after eliminating all duplicate pairs.
    So, the length of $P'$ is $p''$.
    For every $i \in [p'']$, we use $(c_i,b_i)$ to denote the $i$th pair in $P'$, and we use $w_i$ to denote the number of times that this pair appears in the sequence $P$.
    Thus, $\sum_{i=1}^{p''}w_i=p'$.
    Furthermore, the goal of the claim can now be restated as proving that
    \begin{equation}\label{eq: bad event in claim in proof of modified vector system}
        \Pr\left(\left\|\sum_{i=1}^{p''} (w_i \cdot v^{c_i}_{b_i})\right\|_\infty < \frac{p'}{2} + \sqrt{\frac{p'}{200}}\right) \leq \frac{1}{2 \cdot p \cdot (2q)^p}.
    \end{equation}
    For every $i \in [p'']$ and every $j \in [d_0]$, let $v^{c_i}_{b_i}[j]$ denote the $j$'th entry of vector $v^{c_i}_{b_i}$.
    Furthermore, for every $j \in [d_0]$, let $\mathcal{E}_j$ denote the event that the inequality $\sum_{i=1}^{p''}(w_i \cdot v^{c_i}_{b_i}[j]) < \frac{p'}{2} + \sqrt{\frac{p'}{200}}$ holds.
    Then, the probability expression in the left-hand side of~\eqref{eq: bad event in claim in proof of modified vector system} exactly describes the probability that all events $\mathcal{E}_1,\ldots,\mathcal{E}_{d_0}$ occur simultaneously.
    So, it suffices the prove that
    $$\Pr\left(\mathcal{E}_1 \land \ldots \land \mathcal{E}_{d_0}\right) \leq \frac{1}{2 \cdot p \cdot (2q)^p}.$$
    Observe also that by the definition of the vector system and the definition of the events $\mathcal{E}_1,\ldots,\mathcal{E}_{d_0}$, these events are independent.
    Therefore, it suffices to prove that
    $$\Pr\left(\mathcal{E}_1\right) \cdot \ldots \cdot \Pr\left(\mathcal{E}_{d_0}\right) \leq \frac{1}{2 \cdot p \cdot (2q)^p}.$$
    Thus, since $d_0=\roundup{100 \cdot p \log q}$ implies that $\left(1-\frac{1}{10}\right)^{d_0} \leq 2^{-d_0/10} \leq 2^{-10p \log q} \leq \frac{1}{2 \cdot p \cdot (2q)^p}$, it suffices to prove that for every $j \in [d_0]$,
    $$\Pr\left(\mathcal{E}_j\right) \leq (1-\frac{1}{10}).$$
    It thus suffices to show that for any $j \in [d_0]$, 
    \begin{equation}\label{eq: final goal in claim in construction of modified vector system}
        \sum_{i=1}^{p''}(w_i \cdot v^{c_i}_{b_i}[j]) \geq \frac{p'}{2} + \sqrt{\frac{p'}{200}}
    \end{equation}
    holds with probability at least $\frac{1}{10}$.
    Observe that the values $\left\{v^{c_i}_{b_i}[j] \mid i \in [p'']\right\}$ are independent Bernoulli random variables that each have expectation $1/2$.
    Therefore, by Claim~\ref{cl: claim about Bernouli random variables}, the random variable $X = \sum_{i=1}^{p''}(w_i \cdot v^{c_i}_{b_i}[j])$ satisfies
    \begin{equation}\label{eq:result of using probability claim}
        \Pr\left[X \geq \Exp[X] + \frac{\sigma}{\sqrt{50}}\right] \geq \frac{(1 - 1/\sqrt{50})^2}{4} \geq \frac{1}{10},
    \end{equation}
    where $\sigma$ denote the standard deviation of $X$.
    Furthermore, it is not hard to see that the standard deviation of $X$ is $\sigma=\sqrt{\frac{w_1^2}{4}+\ldots+\frac{w_{p''}^2}{4}} \geq \sqrt{\frac{w_1}{4}+\ldots+\frac{w_{p''}}{4}} = \sqrt{\frac{p'}{4}}$.
    Plugging this into~\eqref{eq:result of using probability claim}, we get that $\Pr[X \geq \Exp[X] + \sqrt{\frac{p'}{200}}] \geq \frac{1}{10}$.
    As $\Exp[X] = \frac{\sum_{i=1}^{p''}w_i}{2} = \frac{p'}{2}$, we conclude that~\eqref{eq: final goal in claim in construction of modified vector system} holds with probability at least $\frac{1}{10}$, as we needed to prove.
    This concludes the proof of Claim~\ref{cl: claim in construction of modified vector system}.
\end{proof}

\noindent This concludes the proof of Claim~\ref{clm: constructing modified vector systems}.

\end{proof}
\subsection{Proof of Claim~\ref{clm : reduction for ell-infinity}}\label{sec : reduction for ell-infinity}

We present a reduction that given a 3SAT formula $\varphi$ of size $m$ and an integer parameter $k \geq 2$, produces an $\ell_{\infty}$-Shortest Path instance $\mathcal{I}$ of size $2^{O(k \log k \log m)}$ such that 
\begin{itemize}
\item if $\varphi$ is satisfiable then $\mathcal{I}$ has a solution of $\ell_{\infty}$-cost at most $1$; and,
\item if $\varphi$ is unsatisfiable, then every solution to $\mathcal{I}$ has $\ell_{\infty}$-cost at least ${k^2}/{900}$.
\end{itemize}


At the high level, the reduction is very similar to that in Section~\ref{sec:reduction-l2}. It first constructs a depth-$1$ series-parallel graph $G$ and then uses the tensoring technique from Section~\ref{sec: l2-general-reduction} to obtain $G^{\tensor k}$  with appropriately defined edge costs.

We now present the reduction formally.

\paragraph{Preliminary steps}
First, we apply the reduction from Theorem~\ref{thm:hardness-r-hypergraph-label-cover} with parameters $r=2$ and $\varepsilon = \frac{1}{8k^4}$ to transform the input 3SAT formula $\varphi$ into an instance $\mathcal{H}$ of the $2$-Hypergraph Label Cover problem such that: if $\varphi$ is satisfiable then $\mathcal{H}$ is satisfiable; and, if $\varphi$ is unsatisfiable, then every assignment for $\mathcal{H}$ weakly satisfies at most an $\varepsilon$-fraction of the hyperedges.
By Theorem~\ref{thm:hardness-r-hypergraph-label-cover}, the instance $\mathcal{H}$ contains at most $m^{O\left(\log \frac{2}{\varepsilon}\right)}$ vertices and hyperedges, and at most $2^{O\left(\log \frac{2}{\varepsilon}\right)}$ labels and colors.
The running time of this reduction from Theorem~\ref{thm:hardness-r-hypergraph-label-cover} is $m^{O\left(\log \frac{2}{\varepsilon}\right)} = m^{O(\log k)}$.
Let $q=2^{O\left(\log \frac{2}{\varepsilon}\right)}$ denote the number of colors in $\mathcal{H}$, and define $p=2k^2$.

Next, we run the algorithm from Claim~\ref{clm: constructing modified vector systems} with parameters $p$ and $q$ as defined in the previous paragraph, to construct a $(2,q,p)$-modified vector system $\calS = \{v_i^c: i\in [2], c\in [q]\}$ of dimension $d_0=O(p \log q)$.
Observe that $q=2^{O\left(\log \frac{2}{\varepsilon}\right)} = 2^{O\left(\log \operatorname{poly}(k)\right)} = \operatorname{poly}(k)$ and $p=2k^2 = \operatorname{poly}(k)$.
Hence, by Claim~\ref{clm: constructing modified vector systems}, the running time of this step is $2^{\operatorname{poly}(k)}$.

\paragraph{Constructing base graph $G$} Now, we transform the $2$-Hypergraph Label Cover instance $\mathcal{H}$ into an instance $\hat{\mathcal{I}}$ of the $\ell_{\infty}$-Shortest Path problem using a construction nearly identical to the one in Section~\ref{subsec: basic reduction}: we construct the same graph $G$ as in Section~\ref{subsec: basic reduction} and the only difference is that we use the modified vector system $\calS$ instead of the one we used before.
In Section~\ref{subsec: proof of no-instance untensored for ell infinity}, we state this construction explicitly and prove the following claim, which states the relevant properties of this construction.

\begin{claim}\label{cl: properties of untensored instance in ell infinity proof}
    The instance $\hat{\mathcal{I}}$ of the $\ell_{\infty}$-Shortest Path problem satisfies the following four properties.
    \begin{enumerate}
        \item\label{prop: size of untensored instance in ell infinity proof} The graph in the instance $\hat{\mathcal{I}}$ contains at most $2^{O(\log k \log m)}$ edges, and the dimension $d$ of the cost vectors in this instance is $2^{O(\log k \log m)}$.
        \item\label{prop: yes-instance untensored in ell infinity proof} If the 3SAT formula $\varphi$ is satisfiable, then the $\ell_{\infty}$-Shortest Path instance $\hat{\mathcal{I}}$ has a solution of $\ell_{\infty}$-cost at most $1$.
        \item\label{prop: no-instance untensored in ell infinity proof} If the 3SAT formula $\varphi$ is unsatisfiable, then for every integer $y < k^2$ and every $y$ solutions $P_1,\ldots,P_{y}$ to the instance $\hat{\mathcal{I}}$, the inequality $\left\|\sum_{i=1}^{y} \cost(P_{i})\right\|_\infty \geq y + \sqrt{\frac{y}{100}}$ holds.
        \item\label{prop: 0-1 costs} The entries of the cost vectors of the edges in $\hat{\mathcal{I}}$ all belong to $\{0,1\}$.
    \end{enumerate}
    
\end{claim}

\paragraph{Tensoring step} We proceed with the presentation of the algorithm. As in Section~\ref{sec: l2-general-reduction}, we construct graph $G^{\otimes k}$. Recall that every edge $e$ of $G^{\otimes k}$ is encoded by a $k$-tuple of edges of $G$: $e = e(e_1,\ldots,e_k)$. For such an edge $e$, we define its edge cost by $c(e) = \bigotimes_{i=1}^{k} c(e_i)$ (similarly to what we did for $p < \infty$).
This concludes the description of the algorithm.

It remains to upper bound the size of the obtained instance $\mathcal{I}$ and the running time of the reduction, and most importantly to prove that $\mathcal{I}$ indeed satisfies the completeness and soundness guarantees from Claim~\ref{clm : reduction for ell-infinity}.

\paragraph{Size Analysis}
We upper bound the size of the obtained instance $\mathcal{I}$.
By Property~\ref{prop: size of untensored instance in ell infinity proof} of Claim~\ref{cl: properties of untensored instance in ell infinity proof}, the number of edges $|E(G)|$ and dimension $d$ of the cost vectors in the instance $\hat{\mathcal{I}}$ is $2^{O(\log k \log m)}$.
Therefore, the number of edges  in the instance $\mathcal{I}$ is $|E(G^{\otimes k})| = |E(G)|^k = 2^{O(k \log k \log m)}$, and the cost vector dimension is $d^k = 2^{O(k \log k \log m)}$, as required.

\paragraph{Running Time Analysis}
Constructing the $2$-Hypergraph Label Cover instance $\mathcal{H}$ takes time $m^{O\left(\log \frac{2}{\varepsilon}\right)} = m^{O(\log k)}$.
Obtaining a $(2,q,p)$-modified vector system takes time $2^{\operatorname{poly}(k)}$.
Thus, as the size of the graph and the dimension of the cost vectors in instance $\hat{\mathcal{I}}$ is $m^{O(\log k)}$ by Property~\ref{prop: size of untensored instance in ell infinity proof} of Claim~\ref{cl: properties of untensored instance in ell infinity proof}, the total runtime required to construct the instance $\hat{\mathcal{I}}$ is $m^{O(\log k)} + 2^{\operatorname{poly}(k)}=m^{\operatorname{poly}(k)}$.
Finally, once the instance $\hat{\mathcal{I}}$ is constructed, transforming it into tensored instance $\mathcal{I}$ takes time that is polynomial in the size of the graph and the dimension of the cost vectors; therefore, this transformation takes time $2^{O(k \log k \log m)} = m^{O(k \log k)}$.
In summary, the total running time of the algorithm is $m^{\operatorname{poly}(k)}$, as desired.

\paragraph{Correctness Analysis}
We now prove the correctness of our reduction by showing that if the 3SAT formula $\varphi$ is satisfiable then $\mathcal{I}$ has a solution of $\ell_{\infty}$-cost at most $1$; and, if $\varphi$ is unsatisfiable, then every solution to $\mathcal{I}$ has $\ell_{\infty}$-cost at least $\frac{k^2}{900}$.

Suppose the 3SAT formula $\varphi$ is satisfiable. Then, by Property~\ref{prop: yes-instance untensored in ell infinity proof} of Claim~\ref{cl: properties of untensored instance in ell infinity proof}, there exists a solution path $P$ for the instance $\hat{\mathcal{I}}$ whose cost $\cost(P)$ satisfies $\left\|\cost(P)\right\|_\infty \leq 1$.
Then, by tensoring the path $P$ we obtain a solution path $P^{\otimes k} = \{e(e_1,\ldots,e_k) \mid e_1,\ldots,e_k \in P\}$ for instance $\mathcal{I}$ with 
$$\cost(P^{\otimes k}) = \sum_{e_1,\ldots,e_k \in P}\left(\bigotimes_{i=1}^k c(e_i)\right) = \left(\sum_{e \in P}c(e)\right)^{\otimes k} = \cost(P)^{\otimes k},$$ which therefore satisfies $\left\|\cost(P^{\otimes k})\right\|_\infty = \left\|\cost(P)\right\|_\infty^k \leq 1$.

Now, suppose the 3SAT formula $\varphi$ is unsatisfiable.
For every $j \leq k$, consider the $\ell_\infty$-Shortest Path instance $\hat{\mathcal{I}}^{\otimes j}$ on the graph $G^{\otimes j}$, where the cost $c(e')$ of each edge $e' \in E(G^{\otimes j})$ is defined as $c(e')=\bigotimes_{i=1}^{j} c(e_i)$, where $(e_1,\ldots,e_j)$ is the $j$-tuple of edges of $G$ that corresponds to the edge $e' \in E(G^{\otimes j})$.
Thus, $\hat{\mathcal{I}}^{\otimes k} = \mathcal{I}$ and $\hat{\mathcal{I}}^{\otimes 1} = \hat{\mathcal{I}}$.
We need the following claim, which we will prove later.

\begin{claim}\label{cl: how tensoring affects cost in ell infinity proof}
    Consider an unsatisfiable 3SAT formula $\varphi$. Let $j\leq k-1$. Assume that the cost of the optimal solution for $\hat{\mathcal{I}}^{\otimes j}$ is at least $y_j\in \mathbb{Z}_{\geq 1}$. Then the cost of the optimal solution to the instance $\hat{\mathcal{I}}^{\otimes (j+1)}$ is at least $\min\{k^2,\left\lceil y_j+\sqrt{\frac{y_j}{100}}\right\rceil\}$.
\end{claim}

Define integers $y_1,\ldots,y_k$ so that $y_1=1$, and, for every $j  \in [k-1]$, the integer $y_{j+1}$ is $y_{j+1} = \min\{k^2,\left\lceil y_j+\sqrt{\frac{y_j}{100}}\right\rceil\}$.
Since $\varphi$ is unsatisfiable, Claim~\ref{cl: how tensoring affects cost in ell infinity proof} implies that 
the cost of the optimal solution to each instance $\hat{\mathcal{I}}^{\otimes j}$ is at least $y_j$ for every $j\in [k]$.

We show that for this sequence $y_j$, we have $y_j \geq \frac{j^2}{900}$ for all $j\leq k$.
Indeed, if $y_{j} \geq \frac{j^2}{900}$ and $j<k$, then $y_{j+1} \geq \frac{j^2}{900}+\sqrt{\frac{j^2}{90000}} = \frac{j^2+3j}{900} \geq \frac{(j+1)^2}{900}$.
Thus, the cost of every solution to $\hat{\mathcal{I}}^{\otimes k} = \mathcal{I}$ has cost at least $y_k \geq \frac{k^2}{900}$, as required.

Now, we prove Claim~\ref{cl: how tensoring affects cost in ell infinity proof}.

\begin{proof}[Proof of Claim~\ref{cl: how tensoring affects cost in ell infinity proof}]
Fix any $j \leq k-1$, assume that the 3SAT formula $\varphi$ is unsatisfiable, and let $y_j$ be an integer as in the statement of the claim.
Consider an $s$-$t$ path $P$ in $G^{\otimes (j+1)}$.
To prove the claim, we  show that the cost of $P$ is at least $\min\{k^2,\left\lceil y_j+\sqrt{\frac{y_j}{100}}\right\rceil\}$.
Let $\hat{y_j} = \min\{k^2-1,y_j\}$. As $k \geq 2$, we have
$$\left\lceil \hat{y_j}+\sqrt{\frac{\hat{y_j}}{100}}\right\rceil \geq \min\{k^2,\left\lceil y_j+\sqrt{\frac{y_j}{100}}\right\rceil\},$$
Therefore, it suffices to prove that $\left\|\cost(P)\right\|_\infty \geq \left\lceil\hat{y_j} + \sqrt{\frac{\hat{y_j}}{100}}\right\rceil$.

The path $P$ in $G^{\otimes (j+1)}$ is specified by an $s$-$t$ path $P'$ in $G^{\otimes j}$, and a collection of paths $P_e$ in $G$ for all edges $e \in P'$.
Using this representation, the cost vector of $P$ can be expressed as
\begin{equation}\label{eq: path in tensored graph as tensor of paths in G, in proof of ell infinity lowerbound}
    \cost(P)=\sum_{e \in P'} c(e) \otimes \cost(P_e).
\end{equation}
Recall that the cost vectors of all edges in $\hat{\mathcal{I}}$ are in $\{0,1\}^d$.
Therefore, the cost vector of each edge in $G^{\otimes j}$ in in $\{0,1\}^{d^j}$.
By the assumption of the claim, the cost of every $s$-$t$ path in $G^{(j)}$ is at least $y_j$. In particular, the cost of $P'$ is at least $y_j$:
$$\left\|\cost(P')\right\|_\infty \geq y_j \geq \hat{y_j}.$$
Thus, there exists a coordinate $z$ and $\hat{y_j}$ distinct edges $e_1,\ldots,e_{\hat{y_j}} \in P'$ such that the $z$'th coordinate of each of the cost vectors $c(e_1),\ldots,c(e_{\hat{y_j}})$ is $1$.
For the remainder of the proof, fix any such $z$ and $e_1,\ldots,e_{\hat{y_j}}$.
Observe that for every $\hat{y_j}$ vectors $u_1,\ldots,u_{\hat{y_j}}$, it holds that 
$$\biggl\|\sum_{i=1}^{\hat{y_j}} c(e_i) \otimes u_i\biggr\|_\infty \geq \biggl\|\sum_{i=1}^{\hat{y_j}} u_i\biggr\|_\infty.$$
In particular, for paths $P_{e_1},\ldots,P_{e_{\hat{y_j}}}$ in $G$  (some of them may be equal) and their cost vectors $\cost(P_{e_1}),\ldots,\cost(P_{e_{\hat{y_j}}}) \in \bbR^d$, we have 
$$\biggl\|\sum_{i=1}^{\hat{y_j}} c(e_i) \otimes \cost(P_{e_i})\biggr\|_\infty \geq \biggl\|\sum_{i=1}^{\hat{y_j}} \cost(P_{e_i})\biggr\|_\infty.$$
It follows from~\eqref{eq: path in tensored graph as tensor of paths in G, in proof of ell infinity lowerbound} that
\begin{equation}\label{eq: cost of path in tensored graph as sum of costs of paths in blocks, in proof of ell infinity lowerbound}
    \left\|\cost(P)\right\|_\infty
    = \biggl\|\sum_{e \in P'} c(e) \otimes \cost(P_e)\biggr\|_\infty
    \geq \biggl\|\sum_{i=1}^{\hat{y_j}} c(e_i) \otimes \cost(P_{e_i})\biggr\|_\infty
    \geq \biggl\|\sum_{i=1}^{\hat{y_j}} \cost(P_{e_i})\biggr\|_\infty,
\end{equation}
Here we used that all cost vectors $c(e) \otimes \cost(P_{e})$ have non-negative coordinates.

We lower bound the costs of paths $P_{e_i}$ using Property~\ref{prop: no-instance untensored in ell infinity proof} of Claim~\ref{cl: properties of untensored instance in ell infinity proof}: $$\Bigl\|\sum_{i=1}^{\hat{y_j}} \cost(P_{i})\Bigr\|_\infty \geq \hat{y_j} + \sqrt{\frac{\hat{y_j}}{100}}.$$

Plugging this into~\eqref{eq: cost of path in tensored graph as sum of costs of paths in blocks, in proof of ell infinity lowerbound}, we get that $\left\|\cost(P)\right\|_\infty \geq \hat{y_j} + \sqrt{{\hat{y_j}}}/{10}$.
As all entries of the cost vector $\cost(P)$ must be integral, the last inequality implies that $\left\|\cost(P)\right\|_\infty \geq \Bigl\lceil\hat{y_j} + \sqrt{{\hat{y_j}}}/{10}\Bigr\rceil$, as we needed to prove.
This concludes the proof of Claim~\ref{cl: how tensoring affects cost in ell infinity proof}.
\end{proof}
To complete the proof of Claim~\ref{clm : reduction for ell-infinity}, it now only remains to prove Claim~\ref{cl: properties of untensored instance in ell infinity proof}.

\subsection{Proof of Claim~\ref{cl: properties of untensored instance in ell infinity proof}}\label{subsec: proof of no-instance untensored for ell infinity}

Since the coordinates of the vectors in the modified vector system $\calS$ are either zeros or ones, the coordinates of the cost vectors for $G$ are also zeros and ones, as stated in Property~\ref{prop: 0-1 costs} of Claim~\ref{cl: properties of untensored instance in ell infinity proof}.
Next, we prove the other three properties of this claim.

\paragraph{Proof of Property~\ref{prop: size of untensored instance in ell infinity proof}.}
Our goal in this paragraph is to prove that the number of edges and the dimension of the cost vectors in $\hat{\mathcal{I}}$ are both upper bounded by $2^{O(\log k \log m)}$. 
The number of blocks in the series parallel graph of instance $\hat{\mathcal{I}}$ is equal to the number of vertices in $\mathcal{H}$, which is $|V|=m^{O\left(\log \frac{2}{\varepsilon}\right)}$. Furthermore, each such block contains $|L| = 2^{O\left(\log \frac{2}{\varepsilon}\right)}$ edges.
So, the total number of edges in $\hat{\mathcal{I}}$ is $|V| \cdot |L| \leq (2m)^{O\left(\log \frac{2}{\varepsilon}\right)} \leq 2^{O(\log k \log m)}$, where we used the definition of $\varepsilon = \frac{1}{8k^4}$ in the last inequality.
It remains to upper bound the dimension of the cost vectors.
Recall that the dimension of the cost vectors in $\hat{\mathcal{I}}$ is $d=d_0 \cdot |\mathcal{E}|$, where $d_0=O(p \log q)$ is the dimension of vectors in the $(2,q,p)$-modified vector system $\calS$, and $|\mathcal{E}|= m^{O\left(\log \frac{2}{\varepsilon}\right)}$ is the number of hyperedges in $\mathcal{H}$.
Furthermore, as $p=2k^2$ and $q=2^{O\left(\log \frac{2}{\varepsilon}\right)}$, we have that $d_0 \ll m^{O\left(\log \frac{2}{\varepsilon}\right)}$.
Thus,
$$d = d_0 \cdot |\mathcal{E}| = m^{O\left(\log \frac{4}{\varepsilon}\right)} = 2^{O\left(\log k \log m\right)}.$$

\paragraph{Proof of Property~\ref{prop: yes-instance untensored in ell infinity proof}.}
Suppose the 3SAT formula $\varphi$ is satisfiable.
Then, by the construction of the $2$-Hypergraph Label Cover instance $\mathcal{H}$, there exists an assignment $\sigma$ for $\mathcal{H}$ that satisfies all the hyperedges.
Consider the solution path $P=\left\{e(u,\sigma(u)) : u \in V\right\}$ of $\hat{\mathcal{I}}$ corresponding to this assignment.
Furthermore, consider any hyperedge $h=(u_1,u_2) \in \mathcal{E}$ and let $c^h(P) \in \bbR^{d_0}$ denote the restriction of the cost of $P$ to the group of coordinates corresponding to $h$.
To prove Property~\ref{prop: yes-instance untensored in ell infinity proof}, it suffices to show that $\left\|c^h(P)\right\|_\infty \leq 1.$
Observe that $c^h(P) = \sum_{e \in P}c^H(e)$.
Furthermore, note that the only edges $e$ for which the vector $c^h(e)$ has non-zero entries are those in the blocks corresponding to the vertices $u_1$ and $u_2$ of $\mathcal{H}$.
Thus, $c^h(P) = c^H(e(u_1,\sigma(u_1))) + c^H(e(u_2,\sigma(u_2)))$.
We show that $\left\|c^H(e(u_1,\sigma(u_1))) + c^H(e(u_2,\sigma(u_2)))\right\|_\infty \leq 1$.
Since the assignment $\sigma$ satisfies the hyperedge $h=(e_1,e_2)$, we get that the colors $\pi_h^{u_1}(\sigma(u_1))$ and $\pi_h^{u_2}(\sigma(u_2))$ are the same.
So, the cost vectors $c^H(e(u_1,\sigma(u_1)))$ and $c^H(e(u_2,\sigma(u_2)))$ are two distinct vectors of the same color from the modified vector system $\calS$.
By the definition of a modified vector system, we get $$\left\|c^H(e(u_1,\sigma(u_1))) + c^H(e(u_2,\sigma(u_2)))\right\|_\infty \leq 1.$$

\paragraph{Proof of Property~\ref{prop: no-instance untensored in ell infinity proof}.}
Suppose that the 3SAT formula $\varphi$ is unsatisfiable.
Consider any integer $y < k^2$ and solutions $P_1,\ldots,P_y$ to the instance $\hat{\mathcal{I}}$, as in the statement of the property.
Our goal is to prove that $\left\|\sum_{i=1}^{y} \cost(P_{i})\right\|_\infty \geq y + \sqrt{\frac{y}{100}}$.
For every $i \in [y]$, there exists an assignment $\sigma_i$ for $\mathcal{H}$ that corresponds to the path $P_i$, such that $P_i=\left\{e(u,\sigma_i(u)) : u \in V\right\}$.
We begin by showing that there must exist a hyperedge $h \in \mathcal{E}$ that is colorful with respect to the assignments  $\sigma_1,\ldots,\sigma_y$. 
Recall that a hyperedge $h \in \mathcal{E}$ is \emph{colorful} with respect to the assignments $\sigma_1,\ldots,\sigma_y$ if  $\pi_h^u(\sigma_i(u)) \neq \pi_h^v(\sigma_j(v))$ for all $i,j\in [p]$ and all distinct $u,v\in h$.
Since $\varphi$ is unsatisfiable, by the construction of $\mathcal{H}$, every assignment to $\mathcal{H}$ weakly satisfies at most an $\varepsilon$-fraction of the hyperedges $h \in \mathcal{E}$.
So, by Lemma~\ref{lem:colorful}, at least a $1-y^2\varepsilon$ fraction of the hyperedges $h \in \mathcal{E}$ are colorful with respect to the assignments $\sigma_1,\ldots,\sigma_y$.
Since $y < k^2 = p/2$ and $\varepsilon = \frac{1}{2p^2}$, we get that $y^2\varepsilon \leq 1/2$; so, by Lemma~\ref{lem:colorful}, there must exist at least one hyperedge $h \in \mathcal{E}$ that is colorful with respect to the assignments $\sigma_1,\ldots,\sigma_y$.
Fix any  such hyperedge $h \in \mathcal{E}$.
To prove Property~\ref{prop: no-instance untensored in ell infinity proof}, it suffices to show that $\left\|\sum_{i=1}^{y} c^h(P_{i})\right\|_\infty \geq y + \sqrt{\frac{y}{100}}$, which we rewrite as 
\begin{equation}\label{eq: goal in proof of no case in untensored ell infinity construction}
    \left\|\sum_{i \in [y], u \in h} c^h(e(u,\sigma_i(u)))\right\|_\infty \geq y + \sqrt{\frac{y}{100}}.
\end{equation}
using that $\sum_{i=1}^{y} c^h(P_{i}) = \sum_{i \in [y], u \in h} c^h(e(u,\sigma_i(u)))$.

Consider the multiset $Z = \left\{ c^h(e(u,\sigma_i(u))) : i \in [y], u \in h \right\}$ of vectors that appear inside the sum in~\eqref{eq: goal in proof of no case in untensored ell infinity construction}. Observe that all such vectors belong to the $(2,q,p)$-modified vector system $\calS = \{v_i^c: i\in [2], c\in [q]\}$.
Next, we show that no two distinct vectors from $Z$ have the same color according to $\calS$.
Let $u_1,u_2$ denote the two vertices in the hyperedge $h$.
Then, $Z$ contains vectors $v_1^c$ with colors $c \in \left\{\pi^{u_1}_h(\sigma_i(u_1)) : i \in [y]\right\}$ vectors $v_2^c$ with colors $c \in \left\{\pi^{u_2}_h(\sigma_i(u_2)) : i \in [y]\right\}$. We now observe that these two sets of colors are disjoint: indeed, this follows from the definition of a colorful hyperedge, as $h$ is colorful with respect to the assignments $\sigma_1,\ldots,\sigma_y$.
We conclude that no two distinct vectors have the same color, as desired.

Observe that the cardinality of $Z$ is $2y \leq 2k^2 = p$.
Since $Z$ is a multiset of vectors from the $(2,q,p)$-modified vector system and no two distinct vectors in $Z$ have the same color, it follows from the definition of a $(2,q,p)$-modified vector system that $\left\|\sum_{i \in [y], u \in h} c^h(e(u,\sigma_i(u)))\right\|_\infty \geq \frac{|Z|}{2} + \sqrt{\frac{|Z|}{200}} = y + \sqrt{\frac{y}{100}}$, as we needed to prove.
This concludes the proof of Property~\ref{prop: no-instance untensored in ell infinity proof}, and thus concludes the proof of Claim~\ref{cl: properties of untensored instance in ell infinity proof}.

\section{Asymptotic Bounds on Higher-Order Bell Numbers}
\label{sec:bounds-on-bell}
In this section, we first prove the asymptotic lower bounds on $\bell_k(p)$ stated in Theorem~\ref{thm:bkp_lower_bound}, and then present the upper bound. The upper bound is essentially due to~\cite{MOT24}, with a minor refinement for the case $k > k_0$: it improves from $\Theta(p \cdot k^{1 - 1/p})$ to $\Theta(p \cdot (k - k_0)^{1 - 1/p})$, which is relevant only in the regime $k = (1 + \delta)k_0$.

Before proceeding with the proof, we justify the inequality  
\[
p\cdot (k - k_0)^{1 - 1/p} = \Theta\big(p \cdot \min(k^{1 - 1/p},\, k - k_0)\big)
\]
from item 2 of Theorem~\ref{thm:bkp_lower_bound}.  
The left-hand side is clearly at most the right-hand side, since  
\[
p\cdot (k - k_0)^{1 - 1/p} \leq p \cdot k^{1 - 1/p} \quad \text{and} \quad p\cdot (k - k_0)^{1 - 1/p} \leq p\cdot(k - k_0).
\]

To show the reverse bound, we consider two cases. If $k \geq 2k_0$, then $k - k_0 \geq k/2$, and thus  
$(k - k_0)^{1 - 1/p} \geq \frac{1}{2} k^{1 - 1/p}$.
If  $k \in [k_0 + 1, 2k_0)$, then  
$(k - k_0)^{1/p} \leq k_0^{1/p} = O(1)$,  
hence  
$(k - k_0)^{1 - 1/p} \geq \Omega(k - k_0)$.

\subsection{Lower Bound}
To prove the lower bound from Theorem~\ref{thm:bkp_lower_bound}, we will introduce some machinery that reduces the problem of bounding $\bell_k(P)$ to constructing a special family of trees. 

\begin{definition}
\label{def:realizable_trees}
A rooted plane tree is a rooted tree, in which the children of every node are arranged from left to right.
We say that a rooted plane tree $T = (V_T,E_T)$ is a $(k,p)$-tree if it has exactly $p$ leaves and all root-leaf paths have length $k$. 
Let $\mathcal{T}_{k,p}$ be the set of all $(k,p)$-trees. 
We denote the root of a  $(k,p)$-tree $T$ by $r_T$.

We will refer to tree nodes at distance $i$ from the root as level-$i$ nodes. In particular, the root is at level $0$ and all leaves are at level $k + 1$. 

Let us say that two $(k,p)$-trees are identical if one can be obtained from the other by renaming vertices (in particular, the order of children of every node is the same in both trees). 

A bijection $\varphi:V(T) \to V(T')$ is a transformation of a $(k,p)$-tree $T$ to a $(k,p)$-tree $T'$ if it maps root $r_T$ to $r_{T'}$ and adjacent vertices to adjacent; a transformation does not have to preserve the order of children. 

Two trees are equivalent if there is a transformation that maps one to the other.
An automorphism of $T$ is a transformation of $T$ to $T$. We denote the set of all transformations of $T$ as $\Aut(T)$.
\end{definition}
We note that if two $(k,p)$-trees are identical, then they are also equivalent; however, the converse need not hold. Specifically, $T$ and $T'$ are equivalent if there exists a transformation that maps $T$ to $T'$, whereas they are identical if such a transformation additionally \textit{preserves the order of the children}.

\begin{definition}\label{def:tree-to-partition}
Let $\mathbb{S}_p$ be the set of all permutations of $[p]$.
Given a $(k,p)$-tree $T$ and permutation $\tau \in \mathbb{S}_p$, we define a $k$-level hierarchical partition $P(T, \tau)$ as follows.
Denote the leaves of $T$ from left to right as $v_1,\dots, v_p$.
Elements $a,b\in [p]$ belong to the same part of the partition at level $i$ if leaves $v_{\tau(a)}$ and $v_{\tau(b)}$ have a common ancestor at level $i$ (in other words, the lowest common ancestor of $v_{\tau(a)}$ and $v_{\tau(b)}$ is at some level $j \geq i$).

For all $T \in \mathcal{T}_{k,p}$, let $\Gamma_T := \{ P(T,\tau) : \tau \in \mathbb{S}_p \}$ be the set of all $k$-level hierarchical partitions of $[p]$ that can be produced with $T$. 
\end{definition}

We now make a few simple observations about $\Gamma_T$ and $\Aut(T)$. 
This first claim states that if two trees are non-equivalent, then they produce different hierarchical partitions. 

\begin{claim}
    If $T_1 \in \mathcal{T}_{k,p}$ and $T_2\in \mathcal{T}_{k,p}$ are non-equivalent then 
    \[
        \Gamma_{T_1} \cap \Gamma_{T_2} = \emptyset.
    \]
\end{claim}

Now observe that $(T,\alpha)$ and $(T, \beta)$ define the same $k$-level partition if only if there is an automorphism of $T$ that maps each $v_i$ to $v_{\beta\alpha^{-1}(i)}$. Accordingly, we obtain the following corollary.
\begin{claim}
    \[
        |\Gamma_{T}| = \frac{p!}{| \Aut(T) |}.
    \]
\end{claim}

Combining the above claims, we get the following lower bound using the mapping $(T, \tau)\mapsto P(T, \tau)$ from Definition \ref{def:tree-to-partition}. 

\begin{corollary}
    \label{cor:bkp_bound}
    Let $\mathcal{T}$ be a set of non-equivalent $(k,p)$-trees. 
    Then
    \[
        \bell_k(p) \geq p! \sum_{T \in \mathcal{T}} | \Aut(T) |^{-1}.
    \]
\end{corollary}
A tree transformation is uniquely specified by the way it reorders the children of nodes. The number of different ways to perform the reordering is $\Tran(T) = \prod_{u\in V(T)} |r_u|!$, where $r_u$ is the number of children of node $u$.
Note that $|\Aut(T)| \leq \Tran(T)$.
\begin{definition}
Let us say that a $(k,p)$-tree is $C$-nice if $\Tran(T)\leq C^p$. 
\end{definition}

\begin{claim}
    \label{claim:bkp_bound_refined}
    Let $\mathcal{T}$ be a set of $C$-nice $(k,p)$-trees, no two of which are identical. 
    Then
    \[
        \bell_k(p)^{1/p} \geq \frac{p}{e\cdot C^2} |\mathcal T|^{1/p}.
    \]
\end{claim}
\begin{proof}
For every tree $T\in \cal T$, there are at most $\Tran(T)$ equivalent to $T$ trees in $\cal T$. Further, since $T$ is $C$-nice, $\Tran(T) \leq C^p$. We conclude that there is a subset of at least $|\mathcal{T}|/C^p$ pairwise non-equivalent trees in $\cal T$. Applying Corollary~\ref{cor:bkp_bound} to this subset and using that $|\Aut(T)| \leq \Tran(T) \leq C^p$, we get that 
$$\bell_k(p) \geq p! \cdot \frac{|{\mathcal{T}}|}{C^p}\cdot \frac{1}{C^p} \geq \left(\frac{p}{e}\right)^p \cdot \frac{|\cal T|}{C^{2p}}.$$
Raising both sides to the power of $1/p$, we get the desired statement.
\end{proof}

Now we prove Lemma~\ref{lem:automorphism-count} that will immediately yield item 1 of Theorem~\ref{thm:bkp_lower_bound} and will also be used in the proof of item~2.
\begin{lemma}\label{lem:automorphism-count}
Let $f_k(p)$ be the minimum of $\Tran(T)$ over all $(k,p)$-trees $T$. Assume that $\log_e^{(k)} p \geq 1$. Then
$$f_k(p)^{1/p} \leq C\log_e^{(k)} p$$
for $C = 10$.
\end{lemma}
\begin{proof}
We prove the statement by induction on $k$. If $k=0$, the desired tree simply consists of a root connected to $p$ leaves. The tree has $p! \leq p^p = (\log_e^{(0)} p)^p$ transformations.

Assuming that the statement holds for all $k' < k$, we prove it for $k$. Let $p'=  \lfloor{p}/{\log_e^{(k)} p}\rfloor$. 
Note that $\exp(p') = \exp\Bigl({\bigl\lfloor\left.{p}\right/{\log^{(k)}_e p}\bigr\rfloor}\Bigr) \geq p$ for $p\geq 2$ and hence,
$$\log_e^{(k-1)} p' \geq \log_e^{(k)} e^{p'} \geq \log_e^{(k)} p \geq 1$$ 
and therefore, the induction hypothesis applies to $p'$ and $k' = k-1$.

First, we construct a $(k-1, p')$-tree $T'$ with $\Tran(T') = f_{k-1}(p')$.
Let $\Delta = p/p'$. Denote $\Delta_+ = \lceil\Delta\rceil$ and $\Delta_- = \lfloor\Delta\rfloor$. 
If $\Delta_+ = \Delta_-$, we let $a_+ = p'$ and $a_-=0$.
If $\Delta_+ = \Delta_- +1$, we define integers $a_+ = p - p' \Delta_-$ and $a_- = p' - a_+$ so that $p = a_+ \Delta_+ + a_-\Delta_-$. 
Now we attach $\Delta_+$ child nodes to (arbitrarily chosen) $a^+$ leaves of $T'$ and $\Delta_-$ child nodes to the other $a_-$ leaves of $T'$. 
We obtain a $(k,p)$-tree $T$. 

Note that every transformation of $T$ induces a transformation of $T'$. Further, for every transformation of $T'$, there are $(\Delta_+!)^{a_+}(\Delta_-!)^{a_-}$ transformations of $T$. 
Now, $(\Delta_+!)^{a_+} \leq \Delta_+^{a_+ \Delta_+}$ and 
$(\Delta_-!)^{a_-} \leq \Delta_-^{a_- \Delta_-} \leq \Delta_+^{a_- \Delta_-}$. Therefore,
\begin{equation}
\label{eq:number-of-extensions-of-Tp}
(\Delta_+!)^{a_+} \cdot (\Delta_-!)^{a_-} \leq
\Delta_+^{a_+ \Delta_+ + a_-\Delta_-} = \Delta_+^{p}.
\end{equation}
Using the induction hypothesis, the definition of $p'$, and \eqref{eq:number-of-extensions-of-Tp}, we get 
$$\Tran(T) \leq \Tran(T') \cdot \Delta_+^{p} \leq f_{k-1}(p') \cdot \Delta_+^{p}\leq (C\log_e^{(k-1)} p)^{p/\log_e^{(k-1)} p}\cdot \Delta_+^{p}$$
and
$$\Tran(T)^{1/p}\leq (C\log_e^{(k-1)} p)^{1/\log_e^{(k-1)} p}\cdot \Delta_+.$$
Note that, $t^{1/t} \leq e^{1/e}$ for $t > 0$ and $1/\log_e^{(k-1)} p = 1/e^{\log_e^{(k)} p} \leq 1/e$ since $\log_e^{(k)} p \geq 1$.
Also, $\Delta_+ < \Delta + 1 \leq p/(p/\log_e^{(k)} p - 1) + 1 = (p/(p - \log_e^{(k)} p) + 1/\log_e^{(k)} p) \log_e^{(k)} p < e \log_e^{(k)} p$.
Therefore,
$$(C\log_e^{(k)} p)^{1/\log_e^{(k)} p} \Delta_+\leq (C\cdot e)^{1/e} \cdot e \cdot \log_e^{(k)}p \leq C\log_e^{(k)}p,$$
where we used that $C =10$ and thus, $(C\cdot e)^{1/e} \cdot e \leq C$.
\end{proof}

We now prove Theorem~\ref{thm:bkp_lower_bound}, item 1. 

\begin{proof}[Proof of Theorem \ref{thm:bkp_lower_bound}, item 1]
Assume that $k \leq k_0$.
From Lemma~\ref{lem:automorphism-count}, we get that there exists a $(k,p)$-tree $T$ with $|\Aut(T)|^{1/p} \leq 10 \log_e^{(k)} p$. Applying Corollary~\ref{cor:bkp_bound} with ${\cal T} = \{T\}$, we get
$$\bell_k(p)^{1/p} \geq \frac{(p!)^{1/p}}{10\log_e^{(k)} p} = \Omega\left(\frac{p}{\log_e^{(k)} p}\right).$$
\end{proof}

While proving item 1 in Theorem \ref{thm:bkp_lower_bound} required constructing only a single $(k,p)$-tree with a small number of automorphisms, proving item 2 will require building a large family of $(k, p)$-trees. We now present a few definitions and claims to describe and analyze this family of trees.

Consider $k'\in [k]$ and $p' \in [p]$, which we will fix later. Assume that $\log_e^{(k')} p' \leq e$. Then for some $k''\leq k'$, $\log_e^{(k'')} p'\in[1,e)$. Therefore, by Lemma~\ref{lem:automorphism-count} $f_{k''}(p') \leq (10 e)^{p'} < 30^{p'}$ and thus there exists a $30$-nice $(k'',p')$-tree $T_0'$. If $k'' < k'$, we create a new root and attach it to the root of $T_0'$ by a path of length $k'-k''$; denote the obtained tree $T_0$. If $k' = k''$, simply let $T_0 = T_0'$. We obtain a 30-nice $(k',p')$-tree $T_0$.

Represent $p$ as a sum of positive integers
$p = q_1 + \dots + q_{p'}$,
where each $q_i$ is either $\lfloor p/p' \rfloor$ or $\lceil p/p' \rceil$.  Set $\Delta = k - k'$.  Suppose that, for each $i$, we are given a family $\mathcal{Q}_i$ of binary $(\Delta,q_i)$-trees in which no two trees are identical.  We construct a family $\mathcal{T}$ of $(k,p)$-trees by selecting a tree $Q_i \in \mathcal{Q}_i$ for every $i \in [p']$ and attaching $Q_i$ to the $i$-th leaf $v_i'$ of $T_0$, identifying its root $r_{Q_i}$ with $v_i'$.

\begin{claim}
Every tree $T\in \mathcal{T}$ is $30$-nice.
\end{claim}
\begin{proof} Consider a tree $T\in \mathcal{T}$, which is obtained by attaching trees $Q_i$ to leaves of $T_0$. Since $T_0$ is $30$-nice, $\Tran(T_0) \leq 30^{p_0}$. Since $Q_i$ is a binary tree with $q_i$ leaves, it has exactly $q_i - 1$ nodes with 2 children each; all other nodes have $0$ or $1$ children. Therefore, $\Tran(Q_i) = 2^{q_i-1}$. We have
$$\Tran(T) = \Tran(T_0)\cdot \prod_{i=1}^{p'} \Tran(Q_i) \leq 30^{p'} \prod_{i=1}^{p'} 2^{q_i - 1} = 30^{p'}\cdot 2^{p - p'} \leq 30^{p}.$$
\end{proof}
Applying Claim~\ref{claim:bkp_bound_refined} to the family $\cal T$, we get the following corollary. 
\begin{corollary}\label{cor:bell-lower-bound-in-terms-of-Q}
    $$\bell_k(p)^{1/p} \geq \Omega(p)\cdot \prod_{i=1}^{p'} |\mathcal{Q}_i|^{1/p}.$$
\end{corollary}
\begin{claim} 
\label{claim:Q-trees}
Assume that $\Delta \leq k$ and $q \leq p$ satisfy $2\lceil \log_2 q\rceil \leq \Delta$. Then there exists a family $\mathcal{Q}$ of binary $(\Delta, q)$-trees (no two of which are identical) of size at least $\Delta^{q-1}/(8q)$.
\end{claim}
\begin{proof}
Let $\Delta_0 = \lceil \log_2 q \rceil$.  
Consider a full binary tree with $q$ leaves of height $\Delta_0$. Create a new root node $r_Q$ and attach it to the original root by an edge. Denote the obtained binary tree of height $\Delta_0 + 1$ by $Q$.

Now we convert $Q$ into a family of $(\Delta, q)$-trees as follows. Choose a label $h(u)\in [\Delta+1]$ for each node $u$ of $Q$ such that
\begin{itemize}
    \item $h(r_{Q}) = 0$
    \item $h(u) = \Delta + 1$ for every leaf $u$ of $Q$
    \item $h(u) > h(v)$ if $u$ is a child of $v$.
\end{itemize}
Next, replace each edge $(u, v)$ with a path of length $|h(u) - h(v)|$. The resulting graph is a $(\Delta, q)$-tree in which each node $u$ resides at level $h(u)$.  
The number of binary trees produced in this way equals the number of possible choices for the labelings $h$; all resulting trees are pairwise non-identical.

It remains to lower-bound the number of possible labelings $h$. Divide $[k]$ into consecutive intervals $I_1, \dots, I_{\Delta_0}$ with $|I_j| \geq \lceil \Delta/3^{\Delta_0 + 1 - j}\rceil$. This is possible, since 
$$\sum_{j=1}^{\Delta_0} \lceil \Delta/3^{\Delta_0 + 1 - j}\rceil = \left\lceil\frac{\Delta}{3}\right\rceil + \left\lceil\frac{\Delta}{9}\right\rceil+ \dots + \left\lceil\frac{\Delta}{3^{\Delta_0}}\right\rceil\leq \Delta_0 +  
\left(\frac{\Delta}{3}+ \frac{\Delta}{9} + \dots+ \frac{\Delta}{3^{\Delta_0}}\right)\leq \Delta_0 + \frac{\Delta}{2} \leq \Delta.
$$
Now we for every vertex $u$ of $Q$, which is not the root or a leaf, we arbitrarily choose label $h(u)$ in $I_j$, where $j$ is the level of $Q$ in which node $j$ lies in. We let $h(r_Q) = 0$ and $h(u) = \Delta + 1$ for every leaf $u$. We obtain a valid labeling $h$. Denote the number of internal nodes of $Q$ at level $j$ by $\lambda_j$. The number of labelings we can get using this approach is
$$\prod_{j=1}^{\Delta_0} \left\lceil \frac{\Delta}{3^{\Delta_0 + 1 - j}}\right\rceil^{\lambda_j}
\geq \frac{\Delta^{\sum_j \lambda_j}}{3^{\sum_j(\Delta_0 + 1 - j)\lambda_j}}.$$

Now, $\sum_j \lambda_j$ is the number of internal nodes of $Q$ other than $r_Q$, which is $q-1$ (here we use that we created a new root $r_Q$ for the original full binary tree). On the other hand, $\lambda_j \leq 2^j$ and thus 
$$\sum_{j=1}^{\Delta_0}(\Delta_0 + 1 - j)\lambda_j \leq \sum_{j=1}^{\Delta_0}(\Delta_0 + 1 - j)2^j = 2(2^{\Delta_0 + 1} - \Delta_0 -2) \leq 4\cdot 2^{\Delta_0} \leq 8 q.$$
We conclude that $|\mathcal T| \geq \nicefrac{\Delta^{q-1}}{8q}$.
\end{proof}

We now finish the proof of Theorem~\ref{thm:bkp_lower_bound}, item 2. 

\begin{proof}[Proof of Theorem \ref{thm:bkp_lower_bound}, item 2]
If $k \geq 2\log_2 p$, let $p' = 1$ and $k' = 0$.
If $k < 2\log_2 p$, choose $k' = k_0 = \log^{*} p - 1$ and the smallest positive integer $p'$ such that $2\lceil \log_2 \lceil p/p'\rceil\rceil\leq k - k'$. Note that $p'/p = O(2^{-\frac{k-k'}{2}})$. In either case $\log_e^{(k')} p' \leq e$, as we assumed above.

Using these parameters $p'$ and $k'$, we construct a tree $T_0$ and define $q_i$. Note that either 
\begin{enumerate}
    \item $q_1 = p$ and $\Delta = k - k_0 = k \geq  2\lceil \log_2 q_1 \rceil$ or 
    \item
$q_i \leq \lceil p/p'\rceil$ and  $\Delta = k - k' \geq 2\lceil \log_2 \lceil p/p'\rceil\rceil \geq 2\lceil \log_2 q_i\rceil$.
\end{enumerate}
In either case,  the condition of Claim~\ref{claim:Q-trees} is met. Applying the claim, we get a family $\mathcal{Q}_i$ of $(k-k', q_i)$-trees for every $i\in [p']$.
Now Corollary~\ref{cor:bell-lower-bound-in-terms-of-Q} yields a lower bound $\bell_k(p)$ 
$$\bell_k(p)^{1/p} \geq \Omega(p) \cdot \prod_{j=1}^{p'} \left(\frac{\Delta^{q_i-1}}{8q_i}\right)^{1/p} \geq 
\Omega(p) \cdot \frac{\Delta^{\frac{p-p'}{p}}}{O(p/p')^{p'/p}} \geq \Omega(p \Delta^{1 - p'/p}).
$$
In case 1, $p' = 1$ and $\Delta = k$, so $\bell_k(p) \geq \Omega(pk^{1-1/p})$, as required.
In case 2, $p'/p \leq O(2^{-\Delta/2})$ and thus $\Delta^{p'/p} = O(1)$. Therefore, $\bell_k(p) \geq \Omega(p (k - k_0))$.
\end{proof}
\subsection{Upper Bound}
Now we present the upper bound. The proof is a very minor modification of that from~\cite{MOT24}. We will need the following claim from~\cite{bell1938iterated}.
\begin{claim}
Let $f_0(x) = \exp(x)$ and $f_{i+1}(x) = \exp(f_i(x) -1)$. Then the exponential generating function for sequence $(\bell_k(i))_{i=0}^\infty$ is given by:
\begin{equation}\label{eq:bell:recurrence}
\sum_{i=0}^\infty \frac{\bell_{k}(i) x^i}{i!} = f_k(x).
\end{equation}
\end{claim}
From \eqref{eq:bell:recurrence}, we get $\bell_k(p) \leq \frac{p!}{x^p} (f_k(x) - 1)$ for every $x > 0$. Approximating $(p!)^{1/p}$ with $\Theta(p)$, we obtain 
\begin{equation}  \label{eq:bell-base-upper}
\bell_k(p)^{1/p} \leq O(p)\frac{(f_k(x)-1)^{1/p}}{x}.
\end{equation}
If $k\leq k_0$, set $x= \log_e^{(k)} p$. Since $f_{i+1}(x) < \exp(f_i(x))$, we get $f_k(x) \leq \exp(p)$. Therefore, $$\bell_k(p)^{1/p} \leq O(p) \frac{e}{\log_e^{(k)} p},$$
as required in item 1 of Theorem~\ref{thm:bkp_lower_bound}.  

If $k > k_0$, set $x = \log_e(1 + \frac{1}{M})= \Omega(\frac1M)$ where integer parameter $M \geq 1$ will be specified later. We now show that for every $i< M$, 
\begin{equation}\label{eq:bell-induction-ub-on-f}
    f_i(x)\leq 1 +  1/(M -i).
\end{equation}
The proof is by induction on $i$:
\begin{itemize}
    \item 
For $i=0$, $f_i(x) = f_0(x) = \exp(x) = 1 + \frac{1}{M}$.
    \item We prove the induction step from $i$ to $i+1$, using the induction hypothesis for $i$ and inequality $e^x \leq 1+ x+x^2$ for $x\in[0,1]$,
\begin{align*}
    f_{i+1}(x) &= \exp(f_{i}(x) - 1) \leq \exp\left(\frac{1}{M-i}\right) \leq 1 + 
\frac{1}{M-i} \cdot \left(1+ \frac{1}{M-i}\right) \\
&= 1+ \frac{1}{M-(i+1)} - \frac{1}{(M-i)^2(M-(i+1))} \leq 1+ \frac{1}{M-(i+1)}.
\end{align*}
\end{itemize}
Therefore, $f_k(x) \leq 1 + \frac{1}{M-(k+1)}$.
From \eqref{eq:bell-base-upper}, we get 
$$\bell_k(p)^{1/p} = O\left(\frac{p\cdot M}{(M-k - 1)^{1/p}}\right).$$ 
Plugging in $M=2k+1$, we get
\begin{equation}\label{eq:bell-part-2-ub1}
\bell_k(p)^{1/p}\leq O\left(pk^{1-1/p}\right)
\end{equation}

Now set $M= k - k_0 + 1$. By \eqref{eq:bell-induction-ub-on-f}, $f_{M-1}(x) \leq 2$. Using that $f_{i+1}(x) \leq \exp(f_i(x))$ and the definition of $\log^*$-function, we get $f_k(x) = f_{M + k_0 - 1}(x) \leq e^p$.
From \eqref{eq:bell-base-upper}, we obtain 
\begin{equation}\label{eq:bell-part-2-ub2}
\bell_k(p)^{1/p} = O\left(p\cdot e \cdot (k - k_0+1)\right) = O(p\cdot (k-k_0)).
\end{equation}
Inequalities~\eqref{eq:bell-part-2-ub1} and \eqref{eq:bell-part-2-ub2} yield the upper bound in part 2 of Theorem~\ref{thm:bkp_lower_bound}.

\bibliographystyle{siamplain}
\bibliography{refs} 

\appendix

\section{Discussion of Feige's \texorpdfstring{$k$}{k}-Prover System and Theorem~\ref{thm:hardness-r-hypergraph-label-cover}}
\label{sec:Feige}
In this section, we discuss the reduction from Theorem~\ref{thm:hardness-r-hypergraph-label-cover}, which is based on Feige's $r$-prover system~\cite{Feige}. In this system, there are $r$ provers $P_1,...,P_r$ and a verifier $V$.
The system is parameterized by an integer $\ell$, and each prover is associated with a distinct $\ell$-bit codeword, such that the Hamming distance between any two such codewords is at least $\ell/3$.
In this system, the verifier and all the provers receive the same 3SAT formula $\phi$ of size $n$.
They all begin by using the same known reduction (see Proposition 2.1.2 of \cite{Feige}) in order to transform the formula into a 3SAT-5 formula $\phi'$ of size $\mathrm{poly}(n)$ with the properties that:
\begin{itemize}
    \item if $\phi$ is satisfiable then $\phi'$ is satisfiable; and
    \item if $\phi$ is unsatisfiable then every assignment to the variables of $\phi'$ satisfies at most a $(1-\epsilon_{3SAT5})$-fraction of the clauses of $\phi'$, where $\epsilon_{3SAT5}$ is some absolute constant.
\end{itemize}
Then, the verifier uses $O(\ell \log n)$ random bits (that are hidden from the provers) to uniformly sample $\ell$ clauses from the formula $\phi'$, as well as sample one variable from each of the sampled clauses.
The verifier sends each prover an $O(\ell \log n)$-bit message that encodes a subset of the sampled clauses and variables, where this subset is determined by the $\ell$-bit codeword assigned to the prover.
Each prover $P_i$ responds by sending the verifier an $O(\ell)$-bit message encoding a truth-assignment to all variables that appear in the message that $P_i$ received.
Based on the prover's responses, the verifier then decides whether to accept or reject.
Feige's algorithm for the verifier guarantees that, if the input 3SAT formula $\phi$ is satisfiable, then the provers will have a strategy that causes the verifier the accept regardless of the randomness. Furthermore, if the input 3SAT formula $\phi$ is unsatisfiable, then, regardless of the strategy used by the provers, the verifier will reject with probability at least $1-r^2 \cdot 2^{-C \cdot \ell}$, where $C$ is some absolute constant that depends only on $\epsilon_{3SAT5}$.
In the instance of $r$-Hypergraph Label Cover constructed by the reduction from \cite{BCMN}, for each $i \in [r]$, the set $V_i$ of vertices correspond to the set of possible messages that the verifier may send to prover $P_i$, and the set of labels corresponds to the set of possible messages that the prover can respond with.
The set of colors corresponds to the set of possible assignments of values to the variables sampled by the verifier, and the set of hyper-edges corresponds to the set of states of the verifier's random string.
Thus, the size of each set $V_i$ is $2^{O(\ell \log n)}$, the size of the set of labels if $2^{O(\ell)}$, the size of the set of colors is $2^{(\ell)}$, and the number of hyper-edges is $2^{O(\ell \log n)}$.
The only difference between the reduction from \cite{BCMN} and the reduction from our Theorem \ref{thm:hardness-r-hypergraph-label-cover}, is that \cite{BCMN} use the parameter $\ell=r$, while we use $\ell = O(\log \frac{r}{\varepsilon})$, and control this value of $\ell$ using the parameter $\varepsilon$.
When using a smaller value of $\ell$, the resulting instance of $r$-Hypergraph Label Cover is small, but the probability of the verifier accepting an unsatisfiable formula $\phi$ is not as small, which means that the fraction of hyperedges that may be weakly satisfied in the resulting $r$-Hypergraph Label Cover instance is not as small.
Consider also that, by hiding a sufficiently large constant in the $O$-notation of $\ell=O(\log \frac{r}{\varepsilon})$, we guarantee the existence of $r$ different $\ell$-bit codewords whose pair-wise Hamming distance is at least $\ell/3$, and a collection of such codewords can be found via a greedy algorithm in time $2^{O(\ell)}$.

\section{Proof of Claim~\ref{cl: claim about Bernouli random variables}}\label{sec:Bernoulli}
Let $\tilde Y_i = Y_i - \frac{1}{2}$. Then
$X - \mu = \sum_{i=1}^n a_i \tilde Y_i$.
Since all $\tilde Y_i$ are symmetric (random variables $\tilde Y_i$ and $-\tilde Y_i$ have the same distribution), so is $X - \mu$. Hence, 
$
\Pr(X - \mu \geq c\sigma) = \frac{1}{2} \Pr(|X - \mu| \geq c \sigma)
$.

Now we use the Paley--Zygmund inequality that states that for every non-negative random variable $Z$ and $\theta \in [0,1]$,
$$\Pr(Z > \theta \Exp[Z]) \geq (1-\theta)^2 \frac{\Exp[Z]^2}{\Exp[Z^2]}.$$
Let $Z = |X - \mu|$ and $\theta = c$. 
Note that $\Exp[Z^2] = \Var[X] = \sigma^2$ and, by Khintchine's inequality, $\Exp[Z]^2 \geq \frac{1}{2}\sigma^2$ \cite{Szarek1976}.
We get
$$\Pr(|X - \mu| \geq c \sigma) \geq (1-c)^2 \frac{\nicefrac{\sigma^2}{2}}{\sigma^2} = 
\frac{(1-c)^2}{2}$$
and thus $\Pr(X - \mu \geq c\sigma) \geq \frac{(1-c)^2}{4}$. 

We note that by choosing an appropriante $p\in [1,2]$, applying the Paley--Zygmund inequality to $Z = |X-\mu|^p$, and then lower bounding $\Exp[Z]$ and upper bounding $\Exp[Z^{2}]$ using Khintchine's Inequality~\cite{Haagerup1981}, one can get stronger lower bounds on $\rho = \Prob{X \geq \mu + c\sigma}$ for specific values of $c > 0$. Namely, the following bounds hold 

$$\rho \geq \max\limits_{p\in[1,1.8474]}\frac{\sqrt{\pi } \left(1- c^p\right)^2}{8\, \Gamma \left(p+\frac{1}{2}\right)},\quad \rho \geq\max\limits_{p\in[1.8475,2]}\frac{\left(1-c^p\right)^2 \Gamma \left(\frac{p+1}{2}\right)^2}{2 \sqrt{\pi }\, \Gamma \left(p+\frac{1}{2}\right)},\quad\text{and}\quad
\rho \geq (1-c^2)^2/6.$$

\typeout{get arXiv to do 4 passes: Label(s) may have changed. Rerun}
\end{document}